\PassOptionsToPackage{colorlinks,linkcolor={blue},citecolor={blue},urlcolor={red},breaklinks=true,final}{hyperref}
\documentclass[acmsmall,screen,anonymous=false,final,nonacm]{acmart}
\usepackage{ifthen}
\newboolean{arxiv}
\setboolean{arxiv}{true} %
\newcommand{\ifarx}[2]{\ifthenelse{\boolean{arxiv}}{#1}{#2}}

\setcopyright{rightsretained}
\acmDOI{10.1145/3704871}
\acmYear{2025}
\acmJournal{PACMPL}
\acmVolume{9}
\acmNumber{POPL}
\acmArticle{35}
\acmMonth{1}
\usepackage{booktabs}   %
\usepackage{subcaption} %

\newcommand{\arP}{\overrightarrow{\Pow}}

\usepackage{stel-common}

\providecommand{\catname}{\mathbf} 
\providecommand{\clsname}{\mathcal}
\providecommand{\oname}[1]{{\operatorname{\mathsf{#1}}}}

\def\defcatname#1{\expandafter\def\csname B#1\endcsname{\catname{#1}}}
\def\defcatnames#1{\ifx#1\defcatnames\else\defcatname#1\expandafter\defcatnames\fi}
\defcatnames ABCDEFGHIJKLMNOPQRSTUVWXYZ\defcatnames

\def\defclsname#1{\expandafter\def\csname C#1\endcsname{\clsname{#1}}}
\def\defclsnames#1{\ifx#1\defclsnames\else\defclsname#1\expandafter\defclsnames\fi}
\defclsnames ABCDEFGHIJKLMNOPQRSTUVWXYZ\defclsnames

\def\defbbname#1{\expandafter\def\csname BB#1\endcsname{{\bm{\mathsf{#1}}}}}
\def\defbbnames#1{\ifx#1\defbbnames\else\defbbname#1\expandafter\defbbnames\fi}
\defbbnames ABCDEFGHIJKLMNOPQRSTUVWXYZ\defbbnames

\def\Set{\catname{Set}}

\providecommand{\argument}{\operatorname{-\!-}}

\DeclareOldFontCommand{\bf}{\normalfont\bfseries}{\mathbf}

\providecommand{\id}{\mathsf{id}}

\providecommand{\comp}{\mathbin{\circ}}

\providecommand{\xto}[1]{\,\xrightarrow{#1}\,}

\providecommand{\To}{\mathrel{\Rightarrow}}			           %

\providecommand{\dar}{\kern-1.2pt\operatorname{\downarrow}}	
\providecommand{\uar}{\kern-1.2pt\operatorname{\uparrow}}	

\providecommand{\fst}{\oname{fst}}
\providecommand{\snd}{\oname{snd}}

\providecommand{\inl}{\oname{inl}}
\providecommand{\inr}{\oname{inr}}

\providecommand{\ev}{\oname{ev}}

\usepackage{stmaryrd}

\providecommand{\pacman}[1]{}					                     %

\newcommand{\undefine}[1]{\let #1\relax}					                       %

\providecommand{\mone}{{\text{\kern.5pt\rmfamily-}\mathsf{\kern-.5pt1}}}

\makeatletter
\def\mfix#1{\oname{#1}\@ifnextchar\bgroup\@mfix{}}	       %
\def\@mfix#1{#1\@ifnextchar\bgroup\mfix{}}			           %
\makeatother

\providecommand{\case}[3]{\mfix{case}{\mathbin{}#1}{of}{#2}{\kern-1pt;}{\mathbin{}#3}}

\DeclareMathSymbol{\mathinvertedexclamationmark}{\mathord}{operators}{'074}
\DeclareMathSymbol{\mathexclamationmark}{\mathord}{operators}{'041}
\makeatletter
\newcommand{\raisedmathinvertedexclamationmark}{%
  \mathord{\mathpalette\raised@mathinvertedexclamationmark\relax}%
}
\newcommand{\raised@mathinvertedexclamationmark}[2]{%
  \raisebox{\depth}{$\m@th#1\mathinvertedexclamationmark$}%
}
\makeatother
\newcommand{\skitermu}{\Lambda_{\mathsf{\tiny CL}}}
\newcommand{\appp}{\mathsf{app}}
\newcommand{\xCL}{\textbf{xCL}\xspace}

\newcommand{\mybar}[3]{%
  \mathrlap{\hspace{#2}\overline{\scalebox{#1}[1]{\phantom{\ensuremath{#3}}}}}\ensuremath{#3}
}

\newcommand{\barB}{\mybar{0.6}{1.5pt}{B}}

\newcommand{\barSigma}{\mybar{0.9}{0pt}{\Sigma}}

\newcommand{\RelCat}[2][]{\mathbf{Rel}_{#1}(#2)}

\newcommand{\cprd}{\lesssim^{\raisebox{-1pt}{\scriptsize $O$}}}

\newcommand{\contpre}{\lesssim^{\mathsf{ctx}}}
\newcommand{\appsim}{\lesssim^{\mathsf{app}}}
\newcommand{\logrel}{\mathcal{L}}

\newcommand{\cbpvcat}{(\Set^{\fset/{\Tyv}})^{\Tyl}}

\newcommand{\dom}{{\mathrm{dom}}}
\newcommand{\up}{\oname{up}}

\newcommand{\Pt}{V}

\newcommand{\Pow}{\mathcal{P}}
\newcommand{\xTo}{\xRightarrow}

\newcommand{\qqand}{\qquad\text{and}\qquad}

\newcommand{\dleq}[1]{{\rotatebox{#1}{$\preceq$}}}
\newcommand{\dgeq}[1]{{\rotatebox{#1}{$\succeq$}}}

\newcommand{\under}[1]{\lvert#1\rvert}

\renewcommand{\L}{\mathcal{L}}

\newcommand{\var}{\mathsf{var}}

\newcommand{\Sigmas}{\Sigma^{\star}}

\newcommand{\ar}{\mathsf{ar}}

\newcommand{\seq}{\subseteq}
\newcommand{\ol}{\overline}

\newcommand{\outl}{\mathsf{outl}}
\newcommand{\outr}{\mathsf{outr}}

\providecommand{\C}{}

\renewcommand{\C}{{\mathbb{C}}}

\renewcommand{\id}{{\mathsf{id}}}
\renewcommand{\Nat}{\mathds{N}}

\newcommand{\f}{\oname{f}}

\newcommand{\takeout}[1]{\empty}

\newcommand{\ini}{\iota}

\renewcommand{\rho}{\varrho}

\newcommand{\opp}{\mathsf{op}}

\newcommand{\pullbackangle}[2][]{\arrow[phantom,to path={
                     -- ($ (\tikztostart)!1cm!#2:([xshift=8cm]\tikztostart) $)
                        node[anchor=west,pos=0.0,rotate=#2,
                        inner xsep = 0]
                        {\begin{tikzpicture}[minimum
                        height=1mm,baseline=0,#1]
    \draw[-] (0,0) -- (.5em,.5em) -- (0,1em);
                        \end{tikzpicture}}}]{}}

\usepackage{proof}
\renewcommand{\inference}[2]{\infer{~#2~}{~#1~}}

\usepackage{accents}

\usepackage[capitalize,noabbrev,nameinlink]{cleveref} %

\usepackage{enumitem}
\setlist[enumerate,1]{label=(\arabic*),font=\normalfont,align=left,leftmargin=0pt,labelindent=0pt,listparindent=\parindent,labelwidth=0pt,itemindent=!,topsep=2pt,parsep=0pt,itemsep=2pt,start=1}
\setlist[enumerate,2]{label=(\alph*),font=\normalfont,labelindent=*,leftmargin=*,start=1}
\setlist[itemize]{labelindent=*,leftmargin=*}
\setlist[description]{labelindent=*,leftmargin=*,itemindent=-1 em}

\usepackage{ifdraft}

\renewcommand{\comp}{\cdot}
\renewcommand{\c}{\colon}

\usepackage{seqsplit}
\usepackage{xstring}
\usepackage{xcolor}
\usepackage{wasysym}
\usepackage{microtype}

\newcommand{\vccomp}{\mathbin{\raisebox{1pt}{\scalebox{.6}{$\RIGHTcircle$}}}}
\newcommand{\cvcomp}{\mathbin{\raisebox{1pt}{\scalebox{.6}{$\LEFTcircle$}}}}
\newcommand{\cccomp}{\mathbin{\bullet}}
\newcommand{\vvcomp}{\mathbin{\circ}}


\makeatletter
\newcommand{\pushright}[1]{\ifmeasuring@#1\else\omit\hfill$\displaystyle#1$\fi\ignorespaces}
\newcommand{\pushleft}[1]{\ifmeasuring@#1\else\omit$\displaystyle#1$\hfill\fi\ignorespaces}
\makeatother

\tikzstyle{shiftarr}=[
        rounded corners,%
        to path={--([#1]\tikztostart.center)
                     -- ([#1]\tikztotarget.center) \tikztonodes
                     -- (\tikztotarget)},
]

\tikzset{
    commutative diagrams/.cd,
    arrow style=tikz,
    diagrams={>=stealth},
    row sep=large,
    column sep = huge
}

\usetikzlibrary{decorations.pathmorphing}
\ifdraft{%
  \usepackage{hypcap}
  \setcounter{tocdepth}{2}
  \usepackage{showlabels} 
  
}{}

\usepackage[footnote,marginclue,nomargin]{fixme}
\FXRegisterAuthor{hu}{ahu}{HU} %
\FXRegisterAuthor{sm}{asm}{SM} %
\FXRegisterAuthor{st}{ast}{ST} %
\FXRegisterAuthor{ls}{als}{LS} %
\FXRegisterAuthor{sg}{asg}{SG} %

\usepackage{xspace}

\theoremstyle{definition}
\newtheorem{rem}[theorem]{Remark} %

\usepackage{stackengine}
\stackMath
\newcommand\tsup[2][2]{%
 \def\useanchorwidth{T}%
  \ifnum#1>1%
    \stackon[-1.3ex]{\tsup[\numexpr#1-1\relax]{#2}}{\scalebox{2}[1]{$\mathchar"307E$}\kern-.5pt}%
  \else%
    \stackon[-1ex]{#2}{\scalebox{2}[1]{$\mathchar"307E$}\kern-.5pt}%
  \fi%
}

\usepackage{graphicx,scalerel}
\renewcommand\hat[1]{\hstretch{.71}{\widehat{\hstretch{1.4}{#1}}}}

\numberwithin{equation}{section}

\renewcommand{\xto}[1]{\mathrel{\raisebox{-.75pt}{$\xrightarrow{\;\smash{\raisebox{-1.5pt}{{\scriptsize $#1$}}\;}}$}}}
\renewcommand{\xTo}[1]{\mathrel{\raisebox{-.75pt}{$\xRightarrow{\;\smash{\raisebox{-1.25pt}{{\scriptsize $#1$}}\;}}$}}}

\let\xmpsto=\xmapsto
\renewcommand{\xmapsto}[1]{\xmpsto{~#1~}}

\DeclareFontFamily{U}{mathc}{}
\DeclareFontShape{U}{mathc}{m}{it}{<->s*[1.03] mathc10}{}
\DeclareMathAlphabet{\morph}{U}{mathc}{m}{it}

\renewcommand{\id}{\operatorname{\morph{id}}}
\renewcommand{\inl}{\operatorname{\morph{inl}}}
\renewcommand{\inr}{\operatorname{\morph{inr}}}
\renewcommand{\fst}{\operatorname{\morph{fst}}}
\renewcommand{\snd}{\operatorname{\morph{snd}}}
\renewcommand{\outl}{\operatorname{\morph{outl}}}
\renewcommand{\outr}{\operatorname{\morph{outr}}}

\newcommand{\linl}{\oname{inl}}
\newcommand{\linr}{\oname{inr}}
\newcommand{\lfst}{\oname{fst}}
\newcommand{\lsnd}{\oname{snd}}

\def\monoto{\rightarrowtail}
\renewcommand{\Nat}{\mathbb{N}}

\newcommand{\fset}{\mathbb{F}}

\newcommand{\gcat}{\mathbb{C}}

\newcommand{\mS}{{\mu\Sigma}}

\tikzcdset{scale cd/.style={every label/.append style={scale=#1},
    cells={nodes={scale=#1}}}}

\newcommand{\wSigma}{
	\mathchoice
		{\boldsymbol{\Sigma}\kern-.61em{\boldsymbol\Sigma}}
		{\boldsymbol{\Sigma}\kern-.61em{\boldsymbol\Sigma}}
		{\boldsymbol{\Sigma}\kern-.45em{\boldsymbol\Sigma}}
		{\boldsymbol{\Sigma}\kern-.35em{\boldsymbol\Sigma}}
}

\newcommand{\app}{\,}
\newcommand{\twoc}{\mathbf{2}}

\newcommand{\zero}{\mathsf{v}}
\newcommand{\one}{\mathsf{c}}

\newcommand{\CLFT}{\mathbf{G}}

\newcommand{\skif}{\textbf{xCL}$_{\mathrm{fg}}$\xspace}

\newcommand{\cbpv}{\textbf{CBPV}\xspace}

\newcommand{\wt}{\widetilde}

\theoremstyle{definition}
\newtheorem{notation}[theorem]{Notation}

\newcommand{\iimg}[2]{#1{}^\star[#2]}

\newcommand\copf{\mathord{\amalg}}
\newcommand\return[1]{[#1]}

\newcommand{\arty}[2]{#1 \rightarrowtriangle #2} %

\newcommand{\toexpr}[3]{#1\,\mathsf{to}\,#2\,\mathsf{in}\,#3}
\newcommand{\toexprt}[5]{#1\,\mathsf{to}_{#4}\,#2\,\mathsf{in}_{#5}\,#3}

\newcommand{\pred}[1]{\overset{#1}{\rightarrowtail}}
\newcommand{\utype}{\mathsf{unit}}

\newcommand{\vtype}{\varphi}
\newcommand{\ctype}{\kappa}
\newcommand{\type}{\tau}
\newcommand{\Ty}{\mathsf{Ty}}
\newcommand{\Tyl}{\mathsf{Ty}}
\newcommand{\Tyv}{\Phi}
\newcommand{\Tyc}{\mathsf{K}}

\begin{document}\allowdisplaybreaks

\title[Abstract Operational Methods for Call-by-Push-Value]{Abstract Operational
  Methods for Call-by-Push-Value}         %

\author{Sergey Goncharov}
\orcid{0000-0001-6924-8766}             %
\affiliation{
  \department{School of Computer Science}              %
  \institution{University of Birmingham} %
  \city{Birmingham}
  \country{UK}                    %
}
\email{s.goncharov@bham.ac.uk}          %

\author{Stelios Tsampas}
\orcid{0000-0001-8981-2328}             %
\affiliation{
  \institution{Friedrich-Alexander-Universität Erlangen-Nürnberg}            %
  \city{Erlangen}
  \country{Germany}                    %
}
\email{stelios.tsampas@fau.de}          %

\author{Henning Urbat}
\orcid{0000-0002-3265-7168}             %
\affiliation{
  \institution{Friedrich-Alexander-Universität Erlangen-Nürnberg}            %
  \city{Erlangen}
  \country{Germany}                    %
}
\email{henning.urbat@fau.de}          %

\begin{abstract}
Levy's call-by-push-value is a comprehensive programming paradigm
that combines elements from functional and imperative programming, supports
computational effects and subsumes both
call-by-value and call-by-name evaluation strategies. In the present work, we develop modular methods to reason about
program equivalence in call-by-push-value, and in fine-grain
call-by-value, which is a popular lightweight call-by-value sublanguage of the
former. Our approach is based on the fundamental observation that
presheaf categories of \emph{sorted sets} are suitable universes to model
call-by-(push)-value languages, and that natural, coalgebraic notions of program
equivalence such as applicative similarity and logical relations can be
developed within. Starting from this observation, we formalize fine-grain
call-by-value and call-by-push-value in the \emph{higher-order abstract
  GSOS} framework, reduce their key congruence properties
to simple syntactic conditions by leveraging existing theory and
argue that introducing changes to either language incurs minimal proof overhead.

\end{abstract}

\maketitle

\section{Introduction}

Call-by-push-value (CBPV), the subsuming programming paradigm introduced by~\citet{DBLP:phd/ethos/Levy01,DBLP:journals/siglog/Levy22},
has been gradually capturing the interest of researchers
in programming languages. Conceptually, CBPV is based on the slogan
``a value is, a computation
does''. What sets CBPV apart from other models of computation based
on the $\lambda$-calculus is the broadness of its features. For
example, CBPV is able to model computational effects such as exceptions and
non-determinism under the same roof, and both the call-by-name (CBN) and
call-by-value (CBV) $\lambda$-calculus can be faithfully embedded into CBPV.
Due to the latter fact, CBPV is being considered as a useful, formal
\emph{intermediate language} in compilation
chains~\cite{DBLP:conf/esop/McDermottM19, DBLP:journals/pacmpl/KavvosMLD20,
  DBLP:journals/corr/abs-1805-05400}.

The emergence of CBPV as an intermediate language and as a target for formal
verification has precipitated the need to better underpin its mathematical
foundations and also establish a robust theory of program equivalence.
Up to this point, work on program equivalence in CBPV has
been rather limited~\cite{DBLP:conf/itp/RizkallahGZ18,forstercbpv,
  DBLP:conf/esop/McDermottM19}. It applies to specific instances of CBPV,
requires instantiating existing complex methods from the ground up, and typically
relies on complex compatibility lemmas.
Moreover, widespread operational reasoning methods such as applicative
(bi)simulations~\cite{Abramsky:lazylambda,DBLP:journals/iandc/Howe96,pitts_2011}
have not yet been realized in CBPV. We attribute this to the
inherent complexity of CBPV, as well as the largely empirical and delicate
nature of many of the established operational reasoning methods.

The recently introduced framework of \emph{Higher-Order Mathematical Operational
  Semantics}~\cite{gmstu23}, or \emph{higher-order abstract GSOS}, provides
an abstract, categorical approach to the operational semantics of higher-order
languages that extends the well-known \emph{abstract GSOS} framework by \citet{DBLP:conf/lics/TuriP97}. Notably, higher-order abstract GSOS enabled the
development of Howe's method~\cite{UrbatTsampasEtAl23} and (step-indexed)
logical relations~\cite{gmstu24,gmtu24lics} at a high level of generality and in a
language-independent manner. Thus, for all languages that
can be modeled within the abstract framework, effective sound methods for
reasoning about contextual equivalence can be obtained
with neither the need to develop a proof strategy from scratch nor the reliance on laborious
compatibility lemmas. In principle, these methods appear to be the ideal test-bed
to develop efficient reasoning techniques for CBPV. However,
the extent to which higher-order abstract GSOS is suitable for
call-by-value evaluation has been unclear thus far, and the framework has been applied
only to call-by-name languages.

\paragraph{Contributions}
In this paper, we systematically develop a theory of program equivalence for CBPV languages, as well as the related class of {fine-grain call-by-value (FGCBV)} languages~\cite{LevyPowerEtAl03}, based on the higher-order abstract GSOS framework. Our contribution is two-fold:

First, we demonstrate how to model FGCBV and CBPV languages in the abstract GSOS
framework by presenting their operational rules as (di)natural transformations over suitably sorted presheaves, with the sorts providing an explicit distinction between values and computations. In particular, this establishes higher-order abstract GSOS as an eligible setting for call-by-value semantics.

 Second, building on the existing theory of operational methods in higher-order abstract GSOS \cite{UrbatTsampasEtAl23,gmtu24lics}, we derive notions of applicative similarity and (step-indexed) logical relations for both FGCBV and CBPV and prove them to be sound for the contextual preorder. Thanks to generic soundness results available in the abstract framework, the soundness proof for the above notions boils down to checking a simple condition on the rules of the language. We stress that due to the complexities of FGCBV and CBPV, deriving our soundness results from scratch would be a highly non-trivial and laborious task. Therefore, we regard the approach of our paper as an instructive manifestation of the power of categorical methods in operational semantics.

\paragraph{Related work}
The first formal theory of program equivalence for call-by-push-value was developed by
\cite{DBLP:conf/itp/RizkallahGZ18}. The authors introduce an untyped
version of CBPV and develop an equational theory, essentially the congruence
closure of the reduction relation, that is proven sound for contextual equivalence. They
do so by adapting the method of \emph{eager
  normal form bisimulation} by \citet{DBLP:conf/lics/Lassen05}, which is in turn
based on Howe's method~\cite{DBLP:conf/lics/Howe89,DBLP:journals/iandc/Howe96}. In addition, they formalize their results in Coq. Unlike \citet{DBLP:conf/itp/RizkallahGZ18}, our version
of CBPV is typed and, instead of redeveloping a complex method from scratch, we
make use of existing theory.

\citet{forstercbpv} provide a more extensive formalization of the operational
theory of call-by-push-value in Coq. In particular,
they formalize translations of CBV and CBN into CBPV, prove strong
normalization for a typed, effect-free CBPV and present an equational theory,
which is similar to that of \cite{DBLP:conf/itp/RizkallahGZ18} and proven sound for
contextual equivalence. As
an intermediate step, they utilize the \emph{logical
  equivalence} of \citet{10.5555/1076265}. Their
results are heavily dependent on manual compatibility lemmas and it is unclear
whether they apply to other settings.

Logical relations in the broader context of CBPV have been utilized in the work
of \citet{DBLP:conf/esop/McDermottM19}, in order to prove the correctness of
various translations. In particular, the authors extended CBPV to support
call-by-need evaluation, and developed translations from call-by-name
and call-by-need into this extended language.
Step-indexed logical relations for CBPV were used in \cite{10.1145/3290328} to
develop an operational model of Gradual Type Theory.

\section{Abstract Operational Methods}\label{sec:abstract-op-methods}
In this preliminary section, we give a self-contained overview of the theory of higher-order mathematical operational semantics as developed in previous work. Specifically, we recall from~\cite{gmstu23} how to model higher-order languages and their small-step operational semantics in the categorical framework of higher-order abstract GSOS. Moreover, we explain how to derive on that level of abstraction notions of applicative similarity~\cite{UrbatTsampasEtAl23} and (step-indexed) logical relation~\cite{gmtu24lics} that are sound for the contextual preorder.  While our aim is to apply these methods to complex fine-grain call-by-value languages (\autoref{sec:fgcbv}) and call-by-push-value languages (\autoref{sec:cbpv}), the running example for illustrating the concepts of the present section is a simple untyped combinatory logic with a call-by-name semantics, called \emph{extended combinatory logic} ($\xCL$)~\cite{gmstu23}. It is a variant of the well-known \textbf{SKI} calculus~\cite{10.2307/2370619} and forms a computationally complete fragment of the untyped call-by-name $\lambda$-calculus.

\subsection{Extended Combinatory Logic}\label{sec:unary-ski}
\paragraph{Syntax} The set $\skitermu$ of $\xCL$-terms
is generated by the grammar
\[
  t,s \Coloneqq S \mid K \mid I \mid
  S'(t) \mid K'(t) \mid S''(t,s) \mid \appp(t, s).
\]
The binary operation $\appp$ corresponds to function application; we usually write $t\,s$ for $\appp(t,s)$. The standard combinators (constants) $S$, $K$, $I$ represent the $\lambda$-terms \[S=\lambda x.\,\lambda y.\,\lambda z.\, (x\app z)\app (y\app z),\qquad K=\lambda x.\, \lambda y.\, x, \qquad I=\lambda x.\, x.\] The unary operators
$S'$ and $K'$ capture application of~$S$ and~$K$, respectively, to one
 argument: $S'(t)$ behaves like $S\app t$, and $K'(t)$
behaves like $K\app t$.
Finally, the binary operator $S''$ is meant to
capture application of $S$ to two arguments: $S''(t,s)$ behaves like
$(S\app t)\app s$. In this way, the behaviour of each combinator can
be described in terms of \emph{unary} higher-order functions; for example, the
behaviour of $S$ is that of a function taking a term $t$ to $S'(t)$.
\sgnote{This is more confusing than useful. The behaviour of $t$ is $S'(t)$ with unary $S'$. 
Then what is the behaviour of $S'(t)$? What is unary about it?}

\paragraph{Semantics} The small-step operational semantics of $\xCL$ is given by the inductive rules displayed in 
\Cref{fig:skirules}, where $p,p',q,t$ range over terms in $\skitermu$.
\begin{figure*}[t]
\begin{gather*}
\frac{}{S\xto{t}S'(t)}
\qquad\qquad
\frac{}{S'(p)\xto{t}S''(p,t)}
\qquad\qquad
\frac{}{S''(p,q)\xto{t}(p\app t)\app (q\app t)}
\\[1ex]
\frac{}{K\xto{t}K'(t)}
\qquad\quad
\frac{}{K'(p)\xto{t}p}
\qquad\quad
\frac{}{I\xto{t}t}
\qquad\quad
\frac{p\to p'}{p \app q\to p' \app q}
\qquad\quad
\frac{p\xto{q} p'}{p \app q\to p'}
\end{gather*}
\caption{Operational semantics of the $\xCL$ calculus.}
\label{fig:skirules}
\end{figure*}
The rules specify a labeled transition system 
\begin{equation}\label{eq:xcl-lts}\to~
\subseteq \skitermu \product (\skitermu + \{\_\}) \product \skitermu
\end{equation} where $\{\_\}$ denotes the lack
of a transition label and the set $\skitermu$ of labels coincides
with the state space of the transition system. Unlabeled transitions $p\to p'$
correspond to \emph{reductions} (i.e.\ computation steps) and labeled transitions represent
\emph{higher-order behaviour}: a transition $p\xto{t} p_t$ indicates that the program $p$ acts as a function that outputs $p_t$ on input $t$, where $t$ is itself a program term. For instance, the term $(S\app K)\app I$ evolves as follows for every $t\in\skitermu$:
\[ (S \app K) \app I  \to  S'(K)\app I\to  S''(K,I)\xto{t} (K \app t) \app (I \app
    t)\to K'(t) \app (I \app t)\to t \, \cdots \]
Every $p\in\skitermu$ either admits a single
unlabeled transition $p\to p'$ or a family of labeled transitions
${(p\xto{t} p_t)_{t\in\skitermu}}$; thus, the transition system \eqref{eq:xcl-lts} is deterministic, and it may be presented as a map
\begin{equation}\label{eq:oper-model-xcl}
\gamma_0\colon \skitermu \to \skitermu + \skitermu^{\skitermu}
\end{equation}
given by $\gamma_0(p)=p'$ if $p\to p'$ and $\gamma_0(p)=\lambda t. p_t$ otherwise.  In the first case, we say that $p$ \emph{reduces}, and in the second case that $p$ \emph{terminates}. We shall use the following notation for $p,p',t\in \skitermu$:
\begin{itemize}
\item $p\To p'$ if there exist $n\geq 0$ and $p_0,\cdots,p_n\in \skitermu$ such that $p=p_0\to p_1\to \cdots \to p_n=p'$;
\item $p\xTo{t} p'$ if there exists $p''\in \skitermu$ such that $p\To p''\xto{t} p'$;
\item $p \Downarrow p'$ if $p$ eventually terminates in $p'$, that is, $p\To p'$ and $p'$ terminates; 
\item $p{\Downarrow}$ if $p$ eventually terminates, that is, there exists $p'\in \skitermu$ such that $p\Downarrow p'$.
\end{itemize}

\paragraph{Contextual preorder} A simple and natural approach to relating the behaviour of $\xCL$-programs (and programs of higher-order languages in general) is given by the contextual preorder~\cite{morris}. A \emph{context} $C[\cdot]$ is a $\skitermu$-term with a hole `$\cdot$'; we write $C[p]$ for the outcome of substituting the term $p$ for the hole. For instance, $C = (S\app \cdot)\app I$ is a context, and $C[K]=(S\app K)\app I$.
The \emph{contextual preorder} for $\xCL$ is the relation $\contpre\,\seq\skitermu\times \skitermu$ given by
\begin{equation}\label{eq:contpre-xcl} p\contpre q \quad\text{iff}\quad \forall C[\cdot].\, C[p]{\Downarrow} \implies C[q]{\Downarrow}.\end{equation}
Equivalently, the contextual preorder is the greatest relation $R\seq \skitermu\times \skitermu$ that (i) is adequate for termination (i.e.\ if $R(p,q)$ and $p{\Downarrow}$ then $q{\Downarrow}$), and (ii) forms a congruence, i.e.\ is respected by all operations of the language. Indeed, clearly $\contpre$ is adequate (take the empty context $C=[\cdot]$) and a congruence. Moreover, if $R$ is an adequate congruence and $R(p,q)$, then $R(C[p],C[q])$ for every context $C$ because $R$ is a congruence, hence $C[p]{\Downarrow}$ implies $C[q]{\Downarrow}$ by adequacy of $R$, and so $p\contpre q$.

To prove $p\contpre q$, it thus suffices to come up with an adequate congruence $R$ such that $R(p,q)$. There are two natural candidates for such relations: \emph{applicative similarity}~\cite{Abramsky:lazylambda} and the \emph{step-indexed logical relation}~\cite{DBLP:journals/toplas/AppelM01}.
\begin{definition}\label{def:appsim-xcl}
An \emph{applicative simulation} is a relation $R\seq \skitermu\times \skitermu$ such that, whenever $R(p,q)$,
\[
p\to p' \; \implies\; \exists q'.\, q\To q' \,\wedge\, R(p',q') \qqand
p\xto{t} p'\; \implies\; \exists q'.\, q\xTo{t} q'\,\wedge\, R(p',q').\]
\emph{Applicative similarity} $\appsim \,\seq\, \skitermu\times \skitermu$ is the greatest applicative simulation, that is, the union of all applicative simulations.
\end{definition}
Thus applicative similarity is standard weak similarity on the labeled transition system \eqref{eq:xcl-lts}.
\begin{definition}\label{def:logrel-xcl}
The \emph{step-indexed logical relation} $\logrel\seq \skitermu\times \skitermu$ is defined by $\logrel=\bigcap_{n\in \Nat} \logrel^{n}$, where the relations $\logrel^n\seq \skitermu\times \skitermu$ are given inductively as follows: For all $p,q\in \skitermu$ one has $\logrel^0(p,q)$, and moreover $\logrel^{n+1}(p,q)$ iff $\logrel^n(p,q)$ and
\begin{align*}
p\to p'\; & \implies\; \exists q'.\, q\To q' \,\wedge\, \logrel^n(p',q');\\
p \text{ terminates} \;&\implies\; \exists \ol{q}.\, q \Downarrow \ol{q} \,\wedge\, (\forall d,e.p',q'.\, \logrel^n(d,e) \wedge p\xto{d}p' \wedge \ol{q}\xto{e}q' \implies \logrel^n(p',q')). 
\end{align*}  
\end{definition}
The last clause roughly says that related functions send related inputs to related outputs. In contrast, applicative simulations require that related functions send same inputs to related outputs. The following theorem ensures that both applicative similarity and the step-indexed logical relation yield a sound proof method for the contextual preorder:

\begin{theorem}[Soundness]\label{thm:soundness-xcl}
Both $\appsim$ and $\logrel$ are adequate congruences. Hence, for all $p,q\in \skitermu$,
\[
p \appsim q \; \implies\; p\contpre q\qquad\text{and}\qquad
\logrel(p,q) \; \implies\; p\contpre q.\]
\end{theorem}
While adequacy is easy to verify, the congruence property is non-trivial. In the case of $\appsim$, it is typically established via a version of Howe's method~\cite{DBLP:conf/lics/Howe89,DBLP:journals/iandc/Howe96}. The congruence proof for $\logrel$ is structurally simpler and achieved by induction on the number of steps, but still requires tedious case distinctions along the constructors of the language and their respective operational rules. The abstract approach presented next puts soundness results such as \autoref{thm:soundness-xcl} under the roof of a general categorical framework and in this way significantly reduces the proof burden.  

\subsection{Higher-Order Abstract GSOS}\label{sec:ho-gsos}
The language $\xCL$ exemplifies the familiar style of introducing a higher-order language: the syntax is specified by a grammar that inductively generates the set of program terms, and the small-step operational semantics is given by a transition system on program terms specified by a set of inductive operational rules. The categorical framework of higher-order abstract GSOS provides a high-level perspective on this approach. In the following, we first review the necessary background from category theory~\cite{awodey10, mac2013categories} and then show how to model the syntax, behaviour and operational rules of higher-order languages in the abstract framework.

\paragraph*{Notation} For objects
$X_1, X_2$ of a category $\C$, we write $X_1+X_2$ for the coproduct, $\inl\c X_1\to X_1+X_2$ and
$\inr\c X_2\to X_1+X_2$ for its injections, $[g_1,g_2]\c X_1+X_2\to X$ for the copairing of morphisms $g_i\colon X_i\to X$,
$i=1,2$, and $\nabla=[\id_X,\id_X]\colon X+X\to X$ for the codiagonal. We denote the product by $X_1\times X_2$ and
the pairing of morphisms $f_i\c X\to X_i$, $i=1,2$, by $\langle f_1, f_2\rangle\c
X\to X_1\times X_2$. A \emph{relation} on $X\in \C$ is a subobject of $X\times X$, represented by a monomorphism $\langle \outl_{R}, \outr_{R} \rangle \colon R
\pred{} X \times X$; the projections $\outl_R$ and $\outr_R$ are usually left implicit.
The \emph{coslice category} $V/\gcat$,
where $V\in \gcat$, has as objects all \emph{$V$-pointed objects}, i.e.\ pairs $(X,p_X)$ consisting of an object $X\in \gcat$
and a morphism $p_X\c V\to X$, and a morphism from $(X,p_X)$ to
$(Y,p_Y)$ is a morphism $f\c X\to Y$ of $\gcat$ such that
$p_Y = f\comp p_X$. Finally, we denote by $\Set^{\C}$ the category of (covariant) presheaves over a small category $\C$ and natural transformations.

\paragraph*{Algebras in categories} Algebraic structures admit a natural categorical abstraction in the form of {functor algebras}. Given an endofunctor $\Sigma$ on a category $\C$, a \emph{$\Sigma$-algebra}
is a pair $(A,a)$ consisting of an object~$A$ (the \emph{carrier} of the
algebra) and a morphism $a\colon \Sigma A\to A$ (its \emph{structure}). A
\emph{morphism} from $(A,a)$ to an $\Sigma$-algebra $(B,b)$ is a morphism
$h\colon A\to B$ of~$\gcat$ such that $h\comp a = b\comp \Sigma h$.

A \emph{congruence} on $(A,a)$ is a relation $R\monoto A\times A$ that can be equipped with a $\Sigma$-algebra structure $r\colon \Sigma R\to R$ such that both projections $\outl_R,\outr_R\colon (R,r)\to (A,a)$ are $\Sigma$-algebra morphisms. Note that we do not require congruences to be equivalence relations.

A \emph{free $\Sigma$-algebra} on an object $X$ of $\gcat$ is a
$\Sigma$-algebra $(\Sigma^{\star}X,\iota_X)$ together with a morphism
$\eta_X\c X\to \Sigma^{\star}X$ of~$\gcat$ such that for every algebra $(A,a)$
and every morphism $h\colon X\to A$ of $\gcat$, there exists a unique
$\Sigma$-algebra morphism $h^\star\colon (\Sigma^{\star}X,\iota_X)\to (A,a)$
such that $h=h^\star\comp \eta_X$; the morphism $h^\star$ is called the \emph{free
  extension} of $h$. If the category $\C$ is cocomplete and $\Sigma$ is finitary (i.e.\ preserves directed colimits), then free algebras
exist on every object, and their formation gives rise to a monad
$\Sigma^{\star}\colon \gcat\to \gcat$, the \emph{free monad} generated by~$\Sigma$.
For every $\Sigma$-algebra $(A,a)$ we can derive an Eilenberg-Moore algebra
$\hat{a} \colon \Sigma^{\star} A \to A$ whose structure is the free extension of $\id_A\c A\to A$. We write $(\mu \Sigma,\ini)=(\Sigmas 0,\ini_0)$ for the \emph{initial algebra}, viz.\ the free algebra on the initial object $0$.

The standard instantiation of the above concepts is given by algebras for a
signature. Given a set $S$ of \emph{sorts}, an \emph{$S$-sorted algebraic signature} consists of a set~$\Sigma$
of \emph{operation symbols} and a map $\ar\colon \Sigma\to S^{\star}\times S$
associating to every $\f\in \Sigma$ its \emph{arity}. We write $\f\colon s_1\times\cdots\times s_n\to s$ if $\ar(\f)=(s_1,\ldots,s_n,s)$, and $\f\colon s$ if $n=0$ (in which case $\f$ is called a \emph{constant}). Every
signature~$\Sigma$ induces an endofunctor on the category $\Set^S$ of $S$-sorted sets and $S$-sorted functions, denoted by the
same letter $\Sigma$, defined by $(\Sigma X)_s = \coprod_{\f\colon s_1\cdots s_n\to s} \prod_{i=1}^n X_{s_i}$ for $X\in \Set^S$ and $s\in S$. (Functors of this form are called \emph{polynomial functors}.) An algebra for the functor $\Sigma$ is
precisely an algebra for the signature $\Sigma$, viz.~an $S$-sorted set $A=(A_s)_{s\in S}$ equipped with an operation $\f^A\colon \prod_{i=1}^n A_{s_i}\to A_s$ for every $\f\colon s_1\times\cdots \times s_n\to s$ in $\Sigma$. Morphisms of $\Sigma$-algebras are $S$-sorted
maps respecting all operations. 

 A {congruence} on a $\Sigma$-algebra $A$ is a relation ${R\seq A\times A}$ (i.e.\ a family of relations $R_s\seq A_s\times A_s$, $s\in S$) compatible with all operations of $A$: for each $\f\colon s_1\times\cdots \times s_n\to s$ and elements $x_i,y_i\in A_{s_i}$ such that $R_{s_i}(x_i,y_i)$ ($i=1,\ldots,n$), one has $R_s(\f^A(x_1,\ldots,x_n), \f^A(y_1,\ldots,y_n))$.

Given an $S$-sorted set $X$ of
variables, the free algebra $\Sigmas X$ is the $\Sigma$-algebra of
$\Sigma$-terms with variables from~$X$; more precisely, $(\Sigmas X)_s$ is inductively defined by $X_s\seq (\Sigmas X)_s$ and $\f(t_1,\ldots,t_n)\in (\Sigmas X)_s$ for all ${\f\colon s_1\times\cdots \times s_n\to s}$ and $t_i\in (\Sigmas X)_{s_i}$.
In particular, the initial algebra~$\mu \Sigma=\Sigmas 0$ is
formed by all \emph{closed terms} of the signature. We write $t\colon s$ for $t\in (\mS)_s$. For every
$\Sigma$-algebra $(A,a)$, the induced Eilenberg-Moore algebra
$\hat{a}\colon \Sigmas A \to A$ is given by the map that evaluates terms in~$A$.

\paragraph*{Syntax} In higher-order abstract GSOS, the syntax of a higher-order language is modeled by a finitary endofunctor of the form
\[ \Sigma=V+\Sigma'\colon \C\to \C \]
on a presheaf category\footnote{The higher-order abstract GSOS framework works with abstract categories $\C$; in the present paper, we restrict to presheaf categories for economy of presentation, as this suffices to capture the applications in \autoref{sec:fgcbv} and \ref{sec:cbpv}.} $\C=\Set^{\C_0}$, where $\Sigma'\colon \C\to \C$ and $V\in \C$ is an object of \emph{variables}. It follows that $\Sigma$ generates a free monad $\Sigmas$. In particular, $\Sigma$ has an initial algebra $(\mS,\ini)$, which we think of as the object of programs. The requirement that $\Sigma=V+\Sigma'$ explicitly distinguishes programs that are variables: $\mS$ is a $V$-pointed object with point
\begin{equation}\label{eq:pointed-ms} p_{\mS}\colon V\xto{\inl}V+\Sigma'(\mS) = \Sigma(\mS)\xto{\ini} \mS.\end{equation}\par
\begin{example}
 For $\xCL$, we take the polynomial functor on $\C=\Set$ corresponding to the single-sorted signature $\Sigma=\{\,S/0,K/0,I/0,S'/1,K'/1,S''/2,\appp/2\,\}$,
with arities as indicated, and set $V=\emptyset$ since combinatory logics do not feature variables. The initial $\Sigma$-algebra is carried by the set $\skitermu$ of $\xCL$-terms. For languages with variables and binders (\autoref{sec:cbpv}), we will consider categories $\C$ of presheaves over variable contexts and syntax functors corresponding to binding signatures.
\end{example}

\paragraph*{Behaviour} The type of small-step behaviour exposed by a higher-order language is modeled by a mixed-variance bifunctor
\[ B\colon \C^\opp\times \C\to \C\]
such that the intended operational model of the language forms a \emph{higher-order coalgebra}
\begin{equation}\label{eq:op-model}
\gamma\colon \mS\to B(\mS,\mS)
\end{equation}
on the object $\mS$ of program terms.

\begin{example} For $\xCL$, we choose the behaviour bifunctor $B_0$ on $\Set$ given by $B_0(X,Y)=Y+Y^X$. A higher-order coalgebra $c\colon X\to B(X,X)$ is a deterministic transition system with states $X$ where every state $x$ either has a unique unlabeled transition $x\to x'$ (where $x'=\gamma(x)\in X$) or a unique labeled transition $x\xto{e} x_e$ for every $e\in X$ (where $\gamma(x)\in X^X$ and $x_e=\gamma(x)(e)$). For instance, the transition system \eqref{eq:oper-model-xcl} on $\skitermu$ forms a higher-order coalgebra for $B_0$.
\end{example}

\paragraph{Operational rules} The core idea behind higher-order abstract GSOS is to represent small-step operational rules such as those of \autoref{fig:skirules} as \emph{higher-order GSOS laws}, a form of (di)natural transformation that distributes syntax over higher-order behaviour. Formally, a \emph{($V$-pointed) higher-order GSOS law} of $\Sigma\colon \C\to \C$ over $B\colon \C^\opp\times \C\to \C$ is given by a family of morphisms
  \begin{align}\label{eq:ho-gsos-law}
    \rho_{(X,p_X),Y} \c \Sigma (X \times B(X,Y))\to B(X, \Sigma^\star (X+Y))
  \end{align}
  dinatural in $(X,p_X)\in \Pt/\gcat$ and natural in $Y\in \gcat$. (We write $\rho_{X,Y}$ for $\rho_{(X,p_X),Y}$ if the point $p_X$ is clear from the context.) The intention is that a higher-order GSOS law encodes the operational rules of a given language into a parametrically polymorphic family of functions: given an operator $\f$ of the language and the one-step behaviours of its operands $t_1,\ldots,t_n$, the map $\rho_{X,Y}$ specifies the one-step behaviour of the program $\f(t_1,\ldots,t_n)$, i.e.\ the terms it transitions into next. The (di)naturality of $\rho_{X,Y}$ ensures that the rules are parametrically polymorphic, that is, they do not inspect the structure of their meta-variables; cf.~\cite[Prop.~3.5]{gmstu23}.
\begin{example}\label{ex:ho-gsos-law-xcl}
For $\xCL$, we encode the rules of \autoref{fig:skirules} into a higher-order GSOS law
\[
\rho^0_{X,Y}\colon \Sigma(X\times (Y+Y^X))\to \Sigmas(X+Y)+(\Sigmas(X+Y))^X \qquad (X,Y\in \Set)
\]
where the map $\rho^0_{X,Y}$ is given by
\begin{align*}
S \;& \mapsto\; \lambda x. S'(x) & K \;& \mapsto\; \lambda x. K'(x)\\
S'(x,-) \;& \mapsto\; \lambda x'. S''(x,x') & K'(x,-) \;& \mapsto\; \lambda x'. x \\ 
S''((x,-),(x',-)) \;& \mapsto\; \lambda x''. \appp(\appp(x,x''),\appp(x',x'')) & I \;& \mapsto\; \lambda x.x\\
\appp((x,y),(x',-)) \;& \mapsto\; \appp(y,x') &
\appp((x,f),(x',-)) \;& \mapsto\; f(x'),
\end{align*}
for all $x,x'\in X$, $y\in Y$ and $f\in Y^X$. Note that despite $\xCL$ being a deterministic language, applicative simulations (\autoref{def:appsim-xcl}) involve the inherently \emph{non}deterministic concept of weak transition: a given program $p$ may admit multiple (even infinitely many) weak transitions $p\To p'$. To capture this phenomenon, we
 extend the above law $\rho^0$ to a higher-order GSOS law 
\[
\rho_{X,Y}\colon \Sigma(X\times \Pow(Y+Y^X))\to \Pow(\Sigmas(X+Y)+(\Sigmas(X+Y))^X) \qquad (X,Y\in \Set)
\]
of $\Sigma$ over the ``nondeterministic'' bifunctor $\Pow\cdot B_0$, where $\Pow\colon \Set\to\Set$ is the power set functor. This is achieved by applying the law $\rho^0$ element-wise; for instance, for $x,x'\in X$ and $S\in \Pow(Y+Y^X)$,
\[\appp((x,S),(x',-)) \,\mapsto\; \{\, \appp(y,x') \mid y\in S\cap Y \,\} \cup \{\, f(x') \mid f\in S\cap Y^X \,\}. \]
\end{example}

Every higher-order GSOS law $\rho$ \eqref{eq:ho-gsos-law} induces a canonical \emph{operational model}: the higher-order coalgebra $\gamma\colon \mS\to B(\mS,\mS)$ defined by primitive recursion~\cite[Prop.~2.4.7]{DBLP:books/cu/J2016} as the unique morphism making the diagram in \autoref{fig:gamma} commute. Informally, $\gamma$ is the transition system that runs programs according to the operational rules encoded by the given law $\rho$. 

\begin{example} Since the higher-order GSOS law $\rho^0$ (\autoref{ex:ho-gsos-law-xcl}) simply encodes the rules of $\xCL$, its operational model $\gamma_0\colon \skitermu \to \skitermu + \skitermu^{\skitermu}$ defined by \autoref{fig:gamma} coincides with the transition system determined by the rules in \autoref{fig:skirules}. The operational model $\gamma\colon \skitermu \to \Pow(\skitermu + \skitermu^{\skitermu})$ of $\rho$ is $\gamma_0$ composed with the map $u\mapsto \{u\}$. Thus $\gamma$ essentially coincides with $\gamma_0$; this is unsurprising since $\rho$ is merely a formal modification of $\rho_0$ that does not add any new information.
\end{example} 

\begin{figure*}
\begin{tikzcd}[column sep=3em]
\Sigma(\mS) \ar{rrr}{\iota} \ar{d}[swap]{\Sigma\langle \id, \gamma\rangle} & & & \mS \ar[dashed]{d}{\gamma} \\
\Sigma(\mS\times B(\mS,\mS)) \ar{r}{\rho_{\mS,\mS}} & B(\mS,\Sigmas(\mS+\mS)) \ar{r}{B(\id,\Sigmas \nabla)} & B(\mS,\Sigmas(\mS)) \ar{r}{B(\id,\hat\ini)} & B(\mS,\mS) 
\end{tikzcd}
\caption{Operational model of a higher-order GSOS law}\label{fig:gamma}
\end{figure*}

\subsection{Relation Liftings} In order to model notions of applicative simulation and logical relation in the abstract setting, we consider relation liftings of the underlying syntax and behaviour functors. For that purpose, we first turn relations in $\C$ into a suitable category. A \emph{morphism} from a relation $R
\pred{} X \times X$ to another relation $S
\pred{} Y \times Y$ is given by a morphism $f \colon X \to Y$ in $\C$
such that there exists a (necessarily unique) morphism
$R \to S$ rendering the square below commutative:
\[
\begin{tikzcd}[column sep=5em]
	R \ar[d,tail,"{\langle \outl_{R}, \outr_{R} \rangle}"'] \ar[r,dashed] & S
  \ar[d,tail,"{\langle \outl_{S}, \outr_{S} \rangle}"]\\
	X \times X \ar[r,"f \times f"] & Y \times Y
\end{tikzcd}
\]
We write $\Rel(\C)$ for the category of relations in $\C$ and their morphisms. Since $\C=\Set^{\C_0}$, a relation on $X\in \C$ can be presented as a family of set-theoretic relations $R=(R(C)\seq X(C)\times X(C))_{C\in\C_0}$ such that $R(C)(x,x')$ implies $R(C')(Xf(x),Xf(x'))$ for all $f\colon C\to C'$ in $\C_0$.  For every $X\in \C$, the set $\Rel_X(\C)$ of relations on $X$ forms a complete lattice ordered by componentwise inclusion; we denote this order by $\leq$. Moreover, we denote by $\Delta$ the identity relation $\langle \id,\id\rangle\colon X \monoto X\times X$, and by $R\bullet S$ the (left-to-right) composition of relations $R,S\monoto X\times X$; it is computed componentwise as ordinary composition of relations in $\Set$. 
A \emph{relation lifting} of an endofunctor $\Sigma \colon \C \to \C$ is a functor $\overline \Sigma \colon  \RelCat \C \to \RelCat \C$ making the first diagram below commute, where $\under{-}$ denotes the forgetful functor sending $R \pred{} X \times X$ to $X$. Similarly, a \emph{relation lifting} of a bifunctor $B\colon \C^\opp\times \C\to \C$ is a bifunctor $\barB$ making the second diagram commute.
\begin{equation*}
	\label{eq:liftingRel}
	\begin{tikzcd}
		\RelCat{\C} \ar{d}[swap]{\under{-}}  \ar{r}{\barSigma} & \RelCat{\C} \ar{d}{\under{-}}  \\
		\C \ar{r}{\Sigma}  & \C
	\end{tikzcd}
\quad\quad
\begin{tikzcd}
		\RelCat{\C}^\opp\times \RelCat{\C}
		\ar{d}[swap]{\under{-}^\opp\times \under{-}} \ar{r}{\barB} & \RelCat{\C} \ar{d}{\under{-}}  \\
		\C^\opp \times \C \ar{r}{B}  & \C
	\end{tikzcd}
\end{equation*}

\begin{remark}[Canonical Liftings]
\begin{enumerate}
\item Every endofunctor
$\Sigma\colon \C\to \C$ has a \emph{canonical relation lifting} $\overline \Sigma$, which takes a relation $R
\pred{} X \times X$ to the relation $\barSigma R\monoto \Sigma X\times \Sigma X$ given by the image of the morphism $\langle \Sigma\outl_{R}, \Sigma\outr_{R}\rangle\colon \Sigma R\to \Sigma X\times \Sigma X$, obtained via its (surjective, injective)-factorization.
\item Similarly, every bifunctor ${B\c \C^\opp\times \C\to \C}$ has a \emph{canonical relation
lifting} $\overline B$, which takes two relations $R\monoto X\times X$ and $S\monoto Y\times Y$ to the relation $\overline{B}(R,S) \monoto B(X,Y) \times B(X,Y)$ given by the image of the morphism $u_{R,S}$ in the pullback below, obtained via its (surjective, injective)-factorization:
\begin{equation*}
  \label{eq:liftingpb}
  \begin{tikzcd}[column sep=2em, row sep=.4ex]
    &[-3.25em]& {T_{R,S}}
    \pullbackangle{-45}
    &&&& {B(R,S)} \\
    \\
    {\overline{B}(R,S)} \\
    \\
    && {B(X,Y) \times B(X,Y)} &&&& {B(R,Y) \times B(R,Y)}
    \arrow[two heads, from=1-3, to=3-1]
    \arrow[tail, pos=.8, from=3-1, to=5-3]
    \arrow["v_{R,S}", from=1-3, to=1-7]
    \arrow["u_{R,S}"', from=1-3, to=5-3]
    \arrow["{\langle B(\id,\outl_{S}), B(\id,\outr_{S}) \rangle}", from=1-7, to=5-7]
    \arrow["{B(\outl_{R},\id) \times B(\outr_{R},\id)}", from=5-3, to=5-7]
  \end{tikzcd}
\end{equation*}

\end{enumerate}
\end{remark}

\begin{example}\label{ex:functor-liftings}
\begin{enumerate}
\item\label{ex:functor-liftings-polynomial} For a polynomial functor $\Sigma$ on $\Set$, the canonical lifting $\barSigma$ sends $R\seq X\times X$ to the relation $\barSigma R\seq \Sigma X\times \Sigma X$ relating $u,v\in \Sigma X$ iff $u=\f(x_1,\ldots,x_n)$ and $v=\f(y_1,\ldots,y_n)$ for some $n$-ary operation symbol $\f\in\Sigma$, and $R(x_i,y_i)$ for $i=1,\ldots,n$.
\item\label{ex:functor-liftings-pow} The canonical lifting $\ol{\Pow}$ of the power set functor $\Pow\colon\Set\to \Set$ takes a relation $R\monoto X\times X$ to the \emph{(two-sided) Egli-Milner relation} $\ol{\Pow} R\seq \Pow X\times \Pow X$  defined by
\[ \ol{\Pow}R(A,B) \;\iff\; \forall a\in A.\, \exists b\in B.\, R(a,b) \;\wedge\; \forall b\in B.\, \exists a\in A.\, R(a,b).\]
Taking instead the \emph{left-to-right Egli-Milner relation} $\arP R \seq \Pow X\times \Pow X$ given by
\[ \arP R(A,B) \;\iff\; \forall a\in A.\, \exists b\in B.\, R(a,b)\] 
yields a non-canonical lifting. Note that $\arP R$ is \emph{up-closed}: $\arP R(A,B)$ and $B\seq B'$ implies $\arP R(A,B')$.
\item\label{ex:functor-lifting-homspace} The canonical lifting of the bifunctor $B_0(X,Y)=Y+Y^X$ on $\Set$ sends $(R\seq X\times X,\, S\seq Y\times Y)$ to the relation $\barB_0(R,S)\seq (Y+Y^X)\times (Y+Y^X)$ where $\barB_0(R,S)(u,v)$ iff 
\begin{itemize}
\item either $u,v\in Y$ and $S(u,v)$;
\item or $u,v\in Y^X$ and for all $x,x'\in X$, if $R(x,x')$ then $S(u(x),v(x'))$.
\end{itemize}
The second clause expresses that related functions send related inputs to related outputs. This corresponds precisely to the key requirement of logical relations, and hence such relations liftings will be employed to capture logical relations abstractly. 
\end{enumerate}
\end{example}

\subsection{Abstract Soundness Theorem} Next we introduce our abstract notion of contextual preorder, which is parametric in a preorder of \emph{observations}. Here, a \emph{preorder} in $\C=\Set^{\C_0}$ is relation $R\monoto X\times X$ such that each component $R(C)\seq X(C)\times X(C)$, $C\in \C_0$, is reflexive and transitive.

\begin{definition}[Abstract Contextual Preorder]\label{def:abscontpre}
Given a preorder $O\monoto \mS\times \mS$, the \emph{contextual preorder} $\cprd \,\monoto \mS\times \mS$ is the greatest congruence on the initial algebra $\mS$ contained in $O$.
\end{definition}
The contextual preorder can be constructed as the join (in the complete lattice $\Rel_\mS(\C)$) of all congruences contained in $O$, and it is always a preorder~\cite[Thm.~4.17]{gmtu24lics}.

\begin{example}
By choosing $O= \{\, (p,q)\in \skitermu\times \skitermu \mid p{\Downarrow} \implies q{\Downarrow} \,\}$ we recover the contextual preorder $\contpre$ for $\xCL$.
\end{example}

As suggested by \autoref{ex:functor-liftings}\ref{ex:functor-lifting-homspace}, we model applicative similarity and logical relations using relation liftings of behaviour bifunctors. In the following, for $f,g\colon X\to Y$ in $\C$ and $R\monoto Y\times Y$ let $\iimg{(f\times g)}{R}$ be the preimage of $R$ under $f\times g$, that is, the pullback of $\langle \outl_R,\outr_R\rangle$ along $f\times g$.

\begin{definition}[Abstract Applicative Similarity]
  \label{def:app-sim}
Fix a relation lifting $\barB$ of $B\colon \C^\opp\times \C\to \C$ and coalgebras $\gamma, \widetilde{\gamma}\colon \mS\to B(\mS,\mS)$.
An \emph{applicative simulation}  is a relation $R\monoto \mS\times \mS$ such that
\begin{equation}\label{eq:app-sim} R\leq \iimg{(\gamma\times \widetilde{\gamma})}{\barB(\Delta,R)}. \end{equation} 
\emph{Applicative similarity} $\appsim\,\monoto \mS\times \mS$ is the join (in $\Rel_\mS(\C)$) of all applicative simulations. 
\end{definition}

\begin{rem}
\begin{enumerate}
\item The lifting $\barB$ need not be canonical; it is non-canonical in all our applications.
\item By the Knaster-Tarski theorem, $\appsim$ is the greatest fixed point of the monotone map given by $R\mapsto \iimg{(\gamma\times \widetilde{\gamma})}{\barB(\Delta,R)}$ on the complete lattice $\Rel_\mS(\C)$.
\end{enumerate}
\end{rem}

\begin{definition}[Abstract Step-Indexed Logical Relation]
\label{def:abslogrel}
Fix a relation lifting $\barB$ of $B\colon \C^\opp\times \C\to \C$ and coalgebras $\gamma, \widetilde{\gamma}\colon \mS\to B(\mS,\mS)$. The \emph{step-indexed logical relation} $\logrel\monoto \mS\times \mS$ is given by $\logrel=\bigwedge_{\alpha} \logrel^\alpha$, where $\alpha$ ranges over ordinals and $\logrel^\alpha\monoto \mS\times \mS$ is defined by transfinite induction:
\begin{align*}
\logrel^0 &= \mS\times \mS
& \logrel^{\alpha+1} &= \logrel^\alpha \wedge \iimg{(\gamma\times \widetilde{\gamma})}{\barB(\logrel^{\alpha},\logrel^{\alpha})} \\
\logrel^{\alpha} &= \bigwedge\nolimits_{\beta<\alpha} \logrel^\beta && \text{for limit ordinals $\alpha$}.
\end{align*}
\end{definition}
Note that due to $\C=\Set^{\C_0}$ being a well-powered category, the descending chain $(\logrel^\alpha)$ will eventually stabilize, that is, $\logrel=\logrel^\alpha$ for some ordinal $\alpha$.

\begin{example}\label{ex:weakening}
For $\xCL$, we consider the coalgebras $\gamma,\widetilde{\gamma}\colon \skitermu \to \Pow(\skitermu+\skitermu^{\skitermu})$ where $\gamma$ is the operational model of the language, and $\widetilde{\gamma}$ is the \emph{weak operational model} given by
\[ \widetilde{\gamma}(p) = \{\, p' \mid p\To p' \,\} \cup \{\, \gamma_0(p') \mid p\Downarrow p' \,\}. \]
Moreover, we choose the lifting $\arP\cdot \barB_0$ of the behaviour bifunctor $\Pow\cdot B_0$, where $\arP$ is the left-to-right Egli-Milner lifting and $\barB_0$ is the canonical lifting of $B_0$. Then abstract applicative similarity $\appsim$ and the abstract step-indexed logical relation $\logrel$ instantiate to the concrete notions of \autoref{def:appsim-xcl} and \autoref{def:logrel-xcl}. For $\logrel$, one readily verifies that the descending sequence of $\logrel^\alpha$'s stabilizes after $\omega$ steps. (For nondeterministic languages like in \autoref{sec:cbpv}, one generally needs to go beyond $\omega$.)
\end{example}

\begin{remark}In $\xCL$, the simulation condition \eqref{eq:app-sim} is equivalent to
\begin{equation}\label{eq:app-sim-equiv} R\leq \iimg{(\widetilde{\gamma}\times \widetilde{\gamma})}{\barB(\Delta,R)}. \end{equation}
This follows from the observation that in \autoref{def:appsim-xcl}, the premises $p\to p'$ and $p\xto{t} p'$ of the two simulation conditions may be equivalently replaced by weak transitions $p\To p'$ and $p\xTo{t} p'$, respectively.
In general, a coalgebra $\widetilde{\gamma}\colon \mS\to B(\mS,\mS)$ is called a \emph{weakening} of $\gamma\colon \mS\to B(\mS,\mS)$ if \eqref{eq:app-sim} and \eqref{eq:app-sim-equiv} are equivalent for every relation $R\monoto \mS\times \mS$. 
\end{remark}

For all languages modeled in the higher-order abstract GSOS framework, we have the following general congruence and soundness result for applicative similarity and step-indexed logical relations, see \cite[Cor.~VIII.7]{UrbatTsampasEtAl23} and \cite[Cor.~4.33]{gmtu24lics}:

\begin{theorem}[Abstract Soundness Theorem]\label{thm:congruence-ho-abstract-gsos} 
Fix the following data:
\begin{itemize}
\item a higher-order GSOS law $\rho$ of $\Sigma\colon \C\to \C$ over $B\colon \C^\opp\times \C\to \C$;
\item a relation lifting $\barB\colon \Rel(\C)^\opp\times \Rel(\C)\to \Rel(\C)$ of $B$;
\item the operational model $\gamma\colon \mS\to B(\mS,\mS)$ of $\rho$, and a weakening $\widetilde{\gamma}\colon  \mS\to B(\mS,\mS)$. 
\end{itemize}
If this data satisfies \ref{asm:1}--\ref{asm:6} below, then both $\appsim$ and $\L$ are congruences on the initial algebra $\mS$. Therefore, for every preorder $O\monoto \mS\times \mS$, the following implications hold: 
\[
\appsim\, \leq O  \; \implies\; \appsim\,\leq\, \cprd \qqand
\logrel\, \leq O  \; \implies\; \logrel\,\leq\, \cprd.\]
\end{theorem}

\begin{remark}
For the case of $\logrel$, the conditions of the theorem may be slightly relaxed: the coalgebra $\widetilde{\gamma}$ need not be a weakening, and the assumption \ref{asm:3} can be dropped.
\end{remark}

\begin{example}
We recover the soundness result for $\xCL$ (\autoref{thm:soundness-xcl}) by instantiating the Abstract Soundness Theorem to the higher-order GSOS law $\rho$ of $\Sigma$ over $\Pow\comp B_0$ as in \autoref{ex:ho-gsos-law-xcl}, the relation lifting $\arP\comp \barB_0$ of $\Pow\comp B_0$, and the weak operational model $\widetilde{\gamma}$ as in \autoref{ex:weakening}.
\end{example}

It remains to state the assumptions \ref{asm:1}--\ref{asm:6} on the data of the theorem. We list them below and discuss the underlying intuitions afterwards:

\begin{enumerate}[label=(A\arabic*)]
\item\label{asm:1} $\Sigma$ preserves directed colimits, strong epimorphisms, and monomorphisms.
\item\label{asm:2} Each hom-set $\C(Z,B(X,Y))$ (where $X,Y,Z\in \C$) is equipped with the structure of a preorder $\preceq$ such that, for all $q,q'\colon Z\to B(X,Y)$ and $p\colon Z'\to Z$, if $q\preceq q'$ then $q\comp p\preceq q'\comp p$.
\item\label{asm:3} The relation lifting $\barB$ satisfies
\[ \Delta\leq \barB(\Delta,\Delta) \qqand \barB(R,S)\bullet \barB(\Delta,T) \leq \barB(R,S\bullet T)  \text{ for all $R\monoto X\times X$ and $S,T\monoto Y\times Y$}. \]
\item\label{asm:4} For all relations $R\monoto X\times X$ and $S\monoto Y\times Y$, the relation $\barB(R,S)$ is \emph{up-closed}. This means that for every span $B(X,Y)\xleftarrow{f} Z
  \xrightarrow{g} B(X,Y)$ and every morphism $Z\to \barB(R,S)$ such that the
  left-hand triangle in the first diagram below commutes, and the right-hand triangle
  commutes laxly as indicated, there exists a morphism $Z\to \barB(R,S)$ such that the second
  diagram commutes.
\[
  \begin{tikzcd}
    & \ar[bend right=2em]{dl}[swap]{f} \ar[bend left=2em]{dr}{g} \ar{d} Z & {~}\\
    B(X,Y) & \ar[phantom]{ur}[description, pos=.35, xshift=-10]{\dleq{45}} \ar{l}[swap]{\outl} \barB(R,S) \ar{r}{\outr} & B(X,Y)
  \end{tikzcd}
  \implies
  \begin{tikzcd}
    & \ar[bend right=2em]{dl}[swap]{f} \ar[bend left=2em]{dr}{g} \ar[dashed]{d} Z & {~}\\
    B(X,Y) & \ar{l}[swap]{\outl} \barB(R,S) \ar{r}{\outr} & B(X,Y)
  \end{tikzcd}
\]
Here, $\preceq$ in the first diagram refers to the preorder on $\C(Z,B(X,Y))$ chosen in \ref{asm:2}, and $\outl$ and $\outr$ are the projections of the relation $\barB(R,S)\monoto B(X,Y)\times B(X,Y)$.
\item\label{asm:5} $\rho$ has a relation lifting: For each $R\monoto X\times X$ and $S\monoto Y\times Y$, the component $\rho_{X,Y}$ is a $\Rel(\C)$-morphism from $\barSigma(R\times \barB(R,S))$ to $\barB(R,\barSigma^\star(R+S))$, where $\barSigma$ is the canonical lifting.
\item\label{asm:6} The triple $(\mS,\ini,\widetilde{\gamma})$ forms a \emph{higher-order lax
    $\rho$-bialgebra} (cf.~\cite{DBLP:conf/concur/BonchiPPR15} for the corresponding first-order notion), that is, the diagram below commutes laxly:
  \begin{equation*}\label{eq:lax-bialgebra}
    \begin{tikzcd}[column sep=1.95em]
      \Sigma(\mS) \ar{r}{\ini} \ar{d}[swap]{\Sigma \langle \id, \widetilde{\gamma}\rangle}
      &
      \mS \ar{r}{\widetilde{\gamma}}
      &[2em]
      B(\mS,\mS)
      \\
      \Sigma(\mS\times B(\mS,\mS))
      \ar{r}{\rho_{\mS,\mS}}
      &
      B(\mS,\Sigmas(\mS+\mS))
      \ar[phantom]{u}[description]{\dgeq{-90}}
      \ar{r}{B(\id,\Sigmas\nabla)}
      &
      B(\mS,\Sigmas(\mS)) \ar{u}[swap]{B(\id,\hat{\ini})}
    \end{tikzcd}
  \end{equation*}

\end{enumerate}
Let us further elaborate on the above assumptions:
\begin{enumerate}[label=(A\arabic*)]
\item The conditions on the syntax functor $\Sigma$ ensure that the free monads of both $\Sigma$ and $\barSigma$ exist, and moreover that the latter lifts the former~\cite[Prop.~V.4]{UrbatTsampasEtAl23}:
\[\barSigma^{\star} = \ol{\Sigmas}.\]
Here $\ol{(-)}$ refers to the respective canonical liftings. Note that (A1) holds for all polynomial functors.
\item and (A4) are meant to abstract from the up-closure property of the left-to-right Egli-Milner relation (\autoref{ex:functor-liftings}\ref{ex:functor-liftings-pow}). Specifically, for the behaviour bifunctor $\Pow\comp B_0(X,Y)=\Pow(Y+Y^X)$ for $\xCL$ and its lifting $\arP\cdot \barB_0$, we order $\Set(Z,\Pow\comp B_0(X,Y))$ by pointwise inclusion. Then each relation $\arP\comp \barB_0(R,S)$ is up-closed due to the corresponding property of $\arP$.
\item states that the lifting $\barB$ respects diagonals and composition of relations. This condition enables an abstract version of Howe's method for proving congruence of applicative similarity. \ref{asm:3} it always satisfied for the canonical lifting $\barB$ provided that $B$  preserves weak pullbacks~\cite[Prop.~C.9]{utgms23_arxiv}. 
\setcounter{enumi}{4}
\item can be regarded as a monotonicity condition on the rules represented by $\rho$. For instance, for functors $B$ modelling nondeterministic behaviours and whose relation lifting involves the left-to-right Egli-Milner lifting $\arP$, it entails the absence of rules with negative premises~\cite{fiorepositive}. For the canonical lifting $\barB$, the condition always holds~\cite[Constr.~D.5]{utgms23_arxiv}.
\item is the heart of the Abstract Soundness Theorem. Informally, this condition states that the rules encoded by $\rho$ remain sound when strong transitions (represented by $\gamma$) are replaced by weak ones (represented by $\widetilde{\gamma}$). For instance, consider the two rules for application and their weak versions: 
\begin{align*}
	\inference{t\to t'}{t\app s\to t' \app s} 
\qquad
	\inference{t\xto{s} t'}{t \app s\to t'} 
\qquad
	\inference{t\To t'}{t \app s\To t' \app s}
\qquad
	\inference{t \stackon[.5ex]{~\To~}{\scalebox{.75}{$\scriptsize s$}} t'}{t \app s\To t'}
\end{align*}
The third rule is sound because it emerges via repeated application of the first one. The fourth rule follows from the second and third rule. The weak versions of the other rules of \autoref{fig:skirules} are trivially sound because they are premise-free. Hence the lax bialgebra condition is satisfied for $\xCL$.

Generally, the lax bialgebra condition exposes the language-specific core of soundness results for applicative similarity and logical relations, and as illustrated above, its verification is typically straightforward and amounts to an inspection of the rules of the language.
\end{enumerate}

\section{Fine-grain call-by-value}
\label{sec:fgcbv}

The first step towards realizing call-by-push-value is the explicit distinction
between \emph{values} and \emph{computations}. This is not the only idea behind
call-by-push-value; rather, in this halfway point between call-by-value and
call-by-push-value, one speaks of \emph{fine-grain call-by-value
  (FGCBV)}~\cite{LevyPowerEtAl03}\,\footnote{A language similar to FGCBV
  was independently developed by \citet{lassenthesis}.}. In
this paradigm, computations coincide precisely with those terms that can
$\beta$-reduce, and moreover values can be explicitly coerced to computations
through an analogue of the \emph{return} operator of Moggi's computational
  metalanguage~\cite{DBLP:journals/iandc/Moggi91}.

In this section, we introduce \skif, a fine-grain call-by-value untyped combinatory logic that is similar to $\xCL$ (\autoref{sec:unary-ski}),
and demonstrate how our abstract operational methods instantiate to them. The relatively simple nature of \skif allows us to
focus on the key concepts behind fine-grain call-by-value and its modelling in higher-order abstract GSOS; in particular, since the combinatory logic $\xCL$ does not feature variables, we avoid the largely orthogonal technical  challenges of variable management (e.g.\ binding and substitution). FGCBV languages with variables can be treated using the presheaf-based techniques of \autoref{sec:cbpv}.

The top part of \Cref{fig:fgcbv-rules} inductively defines the sort of values $\CLFT_{\zero}$ and 
the sort of computations~$\CLFT_{\one}$ of \skif.
In this multisorted setting, operations and their arities are decorated by their
sorts, which in this case can be either a value or a computation. The combinators $S$, $K$, $I$ belong
to the sort of values and the auxiliary operators $K'$, $S'$, $S''$ can only have computations as operands, as evidenced by the choice of metavariables
$t$ and $s$. On the other hand, the sort of computations consists of the return
operator $[-]$,
whose argument is a \emph{value}, and four different versions of the
(binary) operation of application, one for each combination of computation-computation
($\argument\cccomp\argument$), value-computation $(\argument\vccomp\argument)$,
computation-value $(\argument\cvcomp\argument)$ and value-value
$(\argument\circ\argument)$. We write $v\app w$ for $v\circ w$. Alternatively, one can think of application as a
single operation whose arguments can belong to either sort.

\begin{figure}
  \begin{flalign*}
		\makebox[0pt][l]{\qquad (Values 			$\CLFT_{\zero}$)}&&\hspace{10em}  v,w,u,\ldots&\Coloneqq I\mid K\mid S\mid K'(t)\mid S'(t)\mid S''(s,t)&\hspace{3cm}\\ 
		\makebox[0pt][l]{\qquad (Computations $\CLFT_{\one}$)} &&   t,s,\ldots&\Coloneqq [v] \mid t \cccomp s \mid v \vccomp s \mid t\cvcomp v\mid v \circ w&\\[-2ex]
	\end{flalign*}
  \begin{gather*}
    \hspace{1cm}\inference{}{S\xto{v}\return{S'(\return{v})}}
    \qquad
    \inference{}{S'(t)\xto{v}\return{S''(t,\return{v})}}
    \qquad
    \inference{}{S''(t,s)\xto{v}(t \cvcomp v)\cccomp(s\cvcomp v)}
    \\[2ex]
    \qquad  
    \inference{}{K\xto{v}\return{K'(\return{v})}}
    \qquad
    \inference{}{K'(t)\xto{v}t}
    \qquad
    \inference{}{I\xto{v}\return{v}}
    \qquad
    \inference{}{\return{v}\to v}
    \qquad
    \inference{t \to t'}{t\cccomp s\to t'\cccomp s}
    \\[2ex]
    \infer{t\cccomp s\to v\vccomp s}{t\to v}
    \qquad
    \infer{t\cvcomp v\to t'\cvcomp v}{t \to t'}
    \qquad
    \infer{t\cvcomp w\to v\app w}{t\to v}
    \qquad
    \infer{v \vccomp s \to v \vccomp s'}{s \to s'}
    \qquad
    \infer{v \vccomp s \to v \app w}{s\to w}
    \qquad
    \infer{v \app w \to t}{v \xto{w} t}
  \end{gather*}
  \caption{Call-by-value operational semantics for the \skif calculus.}
  \label{fig:fgcbv-rules}
\end{figure}

The operational semantics of \skif is also presented in \Cref{fig:fgcbv-rules}, where $s,s',t,t'$ range over computations and $v,w$ over values. Computations admit unlabeled transitions (corresponding to reduction steps), and their target can be either a value or a computation. Values admit labeled transitions, and their target is always a computation. Note that only values appear as labels, which echoes the requirement of call-by-value semantics that functions can only be applied
to values. Note also that the application operators can observe the sort of the conclusion of their operands, which is why there are two rules for three of the application operators. Finally, from an operational perspective, the only
task of the return operator is to expose the inner value to the outside.

\subsection{Modelling \skif in Higher-Order Abstract GSOS}

Our goal is to implement \skif in the higher-order abstract GSOS framework and
leverage existing theory to reason about program equivalences. One key insight
of the present work is that, for languages with multiple sorts such as \skif,
one has to work in a mathematical universe where a ``sort'' is an intrinsic
notion. Specifically, we work with the category $\Set^\twoc=\Set\times\Set$ of
two-sorted sets (see also \cite[\textsection
8]{DBLP:journals/lmcs/HirschowitzL22} for a similar idea). We denote objects of $\Set^\twoc$ as pairs of sets $(X_{\zero}, X_{\one})$, and morphisms $f\colon X\to Y$ as pairs of maps $(f^{\zero} \c X_{\zero} \to Y_{\zero},f^{\one} \c X_{\one} \to Y_{\one})$. The intention is to interpret objects of $\Set^{\twoc}$ as pairs of sets of \emph{values} and \emph{computations}. We sometimes write $x\in X$ instead of $x\in X_{\zero}$ or $x\in X_{\one}$ if the sort is irrelevant. 

\begin{notation}
  We write $\copf \c \Set^{\twoc} \to \Set$ for the \emph{sum} functor, mapping each object
  $X \in \Set^{\twoc}$ to $X_{\zero} + X_{\one}$ and each morphism $f \c X \to
  Y$ to $f^{\zero} + f^{\one} \c \copf X \to \copf Y$.
\end{notation}

We argue that the category $\Set^{\twoc}$ is the suitable mathematical
universe to implement the semantics of \skif and, more generally, call-by-value
languages. First, the syntax of \skif
is given by the $\twoc$-sorted algebraic signature corresponding to the polynomial
functor $\Sigma \c \Set^{\twoc} \to \Set^{\twoc}$ given by
\begin{equation}
  \label{eq:sigmaski}
  \begin{aligned}
    & {\Sigma_{\zero}(X)} = \underbrace{1}_{S} + \underbrace{1}_{K}
      + \underbrace{1}_{I} + \underbrace{X_{\one}}_{S'} + \underbrace{X_{\one}}_{K'}
      + \underbrace{X_{\one} \times X_{\one}}_{S''}, \\
    & {\Sigma_{\one}(X)} = \underbrace{X_{\zero}}_{[\argument]}
      + \underbrace{X_{\one} \times X_{\one}}_{\argument \cccomp \argument} +
      \underbrace{X_{\zero} \times X_{\one}}_{\argument \vccomp \argument} +
      \underbrace{X_{\one} \times X_{\zero}}_{\argument \cvcomp \argument} +
      \underbrace{X_{\zero} \times X_{\zero}}_{\argument \vvcomp \argument}.
  \end{aligned}
\end{equation}
The initial algebra of $\Sigma$ is carried by the $\twoc$-sorted set $\CLFT = (\CLFT_{\zero},\CLFT_{\zero})$ of \skif-terms.
The behaviour of \skif is modeled by the bifunctor $B=\langle B_{\zero}, B_{\one}\rangle\c (\Set^{\twoc})^{\opp} \times \Set^{\twoc} \to \Set^{\twoc}$ given by
\begin{equation}
  \label{eq:behski}
   B_{\zero}(X,Y) = Y_{\one}^{X_{\zero}} \qquad \text{and} \qquad
   B_{\one}(X,Y) = \Pow(\copf Y).
\end{equation}
The component $B_{\zero}$ models the behaviour of values, while $B_{\one}$ that
of computations. Values behave as combinators, meaning as
functions from values (and \emph{only} values) to computations. Computations, on
the other hand, may reduce to either a value $Y_{\zero}$ or a
computation $Y_{\one}$. Similar to the categorical modelling $\xCL$ in \autoref{sec:abstract-op-methods}, despite the language being deterministic, we use the power set functor in the computation component in order to subsequently capture weak transitions. The operational rules in \Cref{fig:fgcbv-rules} are represented by a $0$-pointed higher-order GSOS law \[\rho_{X,Y}\c \Sigma(X
\times B(X,Y)) \to B(X,\Sigma^{\star}(X+Y)) \qquad(X,Y\in \Set^\twoc);\]
here $\times$ and $+$ refer to the product and coproduct in $\Set^2$, which are formed as sortwise products and coproducts in $\Set$. The value and computation components of $\rho_{X,Y}$ are maps
\begin{align*} \rho_{X,Y}^{\zero}&\colon \Sigma_{\zero}(X
\times B(X,Y)) \to (\Sigmas_{\one}(X+Y))^{X_{\zero}} \\
\rho_{X,Y}^{\one}&\colon \Sigma_{\one}(X
\times B(X,Y)) \to \Pow(\Sigmas_{\zero}(X+Y)+ \Sigmas_{\one}(X+Y))
\end{align*}
which are defined as follows for $v,w\in X_{\zero}$, $t,s\in X_{\one}$, $L,D\in \Pow(Y_{\zero}+Y_{\one})$ and $f,g\in Y_{\one}^{X_{\zero}}$:
\begin{equation*}
  \begin{aligned}
    \rho_{X,Y}^{\zero}(S)
    &\;= \lambda v.\,[S'([v])]
    &\rho_{X,Y}^{\zero}(S'(t,L))
    &\;= \lambda v.\,[S''(t,[v])]
    \\
    \rho_{X,Y}^{\zero}(S''((t,L),(s,D)))
    &\;= \lambda v.(t \cvcomp v)\cccomp(s \cvcomp v)
    & \rho_{X,Y}^{\zero}(I)
    &\;= \lambda v.\,[v]
    \\
    \rho_{X,Y}^{\zero}(K)
    &\;= \lambda v.\,[K'([v])]
    &\rho_{X,Y}^{\zero}(K'(t,L))
    &\;= \lambda v.t\\[-1ex]
  \end{aligned}
\end{equation*}
\begin{align*} 
\rho_{X,Y}^{\one}([(v,f)])
    &= \{v\} & \rho_{X,Y}^{\one}((t,L) \cccomp (s,D)) &= \{ t'\cccomp s  \mid t'\in L\cap Y_{\one}\} \cup \{ v \vccomp s \mid v\in L\cap Y_{\zero}\} \\ 
&&\rho_{X,Y}^{\one}((t,L) \cvcomp (w,g)) &= \{ t'\cvcomp w \mid t'\in L\cap Y_{\one}\} \cup \{ v \vvcomp w \mid v\in L\cap Y_{\zero}\}\\
&&\rho_{X,Y}^{\one}((v,f) \vccomp (s,D)) &= \{ v\vccomp s' \mid s'\in D\cap Y_{\one}\} \cup \{ v \vvcomp w \mid w\in D\cap Y_{\zero}\}\\*
&&\rho_{X,Y}^{\one}((v,f) \vvcomp (w,g)) &= \{f(w)\} 
\end{align*}

\begin{rem}
The above law is implicitly generated in two steps similar to the law of $\xCL$ (\autoref{ex:ho-gsos-law-xcl}): one first gives a law $\rho^0$ of $\Sigma$ over the deterministic version $B_0$ of $B$ (dropping the power set functor in the $\one$-component), and then extends $\rho^0$ to the above law $\rho$ of $\Sigma$ over $B$.
\end{rem}

The operational model of $\rho$ is the deterministic ($\twoc$-sorted) transition system
\[\gamma = (\gamma^{\zero},\gamma^{\one}) \c \CLFT \to B(\CLFT,\CLFT)\]
on \skif-terms specified by the rules of \autoref{fig:fgcbv-rules}: for every $v\in \CLFT_{\zero}$ and $t\in \CLFT_{\one}$,
\[
  \gamma^{\zero}(v) = \lambda u. t_{u} \iff v \xto{u} t_{u}, \quad
  \gamma^{\one}(t) = \{t'\} \iff t \to t' \quad \text{and} \quad
  \gamma^{\one}(t) = \{v\} \iff t \to v.
\]
Having defined the categorical semantics of \skif, we are now able to instantiate the abstract operational reasoning techniques presented in \autoref{sec:abstract-op-methods}.

\subsection{Operational Methods for \skif}
\label{sec:fgequiv}
Let us first introduce a suitable notion of contextual preorder for \skif. A \emph{$\twoc$-sorted context} $C[\cdot]$ is a $\Sigma$-term with a hole (i.e.\ a term $C$ in a single variable that occurs at most once), and as before we write $C[p]$ for the term emerging by substitution of a term $p$ of compatible sort for the hole. The \emph{contextual preorder} is the $\twoc$-sorted relation $\contpre \,\seq \CLFT\times \CLFT$ given by
\begin{equation} p\contpre q \quad\text{iff}\quad \forall C[\cdot].\, C[p]{\Downarrow} \implies C[q]{\Downarrow},\end{equation}
where $p,q$ are terms of the same sort, $C$ ranges over contexts whose hole is of that sort, and $p{\Downarrow}$ means that either $p$ is a value or $p$ is a computation that eventually reduces to a value. This corresponds to the instantiation of \autoref{def:abscontpre} to the preorder $O\seq \CLFT\times \CLFT$ given by
\begin{equation}\label{eq:o-xclfg}
 O_{\zero} = \CLFT_{\zero} \times \CLFT_{\zero} \qqand O_{\one} = \{ (t,s) \mid t{\Downarrow} \implies s{\Downarrow}\}. \end{equation}

Recall that both abstract applicative similarity (\autoref{def:app-sim}) and the abstract step-indexed logical relation (\autoref{def:abslogrel}) are parametric in a choice of weakening of the operational model $\gamma$ and in the choice of a relation lifting of the behaviour bifunctor $B$. In the present setting, we choose the weakening $\wt{\gamma} \c \CLFT \to B(\CLFT,\CLFT)$ with components $\wt{\gamma}^{\zero} \c \CLFT_{\zero} \to
\CLFT_{\one}^{\CLFT_{\zero}}$ and $\wt{\gamma}^{\one} \c \CLFT_{\one} \to
\Pow(\copf\CLFT)$ given by \begin{equation}
  \wt{\gamma}^{\zero} = \gamma^{\zero} \qquad \text{and} \qquad
  \wt{\gamma}^{\one}(t) = \{e \in \copf\CLFT \mid t \To e\},
\end{equation}
where $\To$ is the reflexive transitive closure of the one-step transition relation ${\to} \subseteq \CLFT_{\one} \times \copf\CLFT$. Moreover, we take the following relation lifting $\barB\colon \Rel(\Set^\twoc)^\opp\times \Rel(\Set^\twoc)\to \Rel(\Set^\twoc)$:
for every $R \pred{} X \times X$ and $S
\pred{} Y \times Y$, the relation $\barB(R,S) \pred{} B(X,Y) \times B(X,Y)$ is given by the components $\barB_{\zero}(R,S) \subseteq Y_{\one}^{X_{\zero}} \times
Y_{\one}^{X_{\zero}}$ and $\barB_{\one}(R,S) \subseteq \Pow(Y_{\zero} +
Y_{\one}) \times \Pow(Y_{\zero} + Y_{\one})$ where
\begin{equation}
  \begin{aligned}
    & \barB_{\zero}(R,S) = \{\,(f,g)
      \mid \forall v,w.\, R_{\zero}(v,w) \implies S_{\one}(f(v),g(w))\,\}. \\
    & \barB_{\one}(R,S) = \arP(S_{\zero} + S_{\one}),\text{  where $\arP$ is the left-to-right Egli-Milner relation lifting.}
  \end{aligned}
\end{equation}
By instantiating \autoref{def:app-sim} to this data, we obtain the following notion:
\begin{definition}[Applicative similarity for \skif]
  \label{def:fgapplsim}
  An \emph{applicative simulation} is a $\twoc$-sorted relation $R\seq \CLFT\times \CLFT$ satisfying the following conditions for all $u,v,w\in \CLFT_{\zero}$ and $t,t',s\in \CLFT_{\one}$:
  \begin{enumerate}
  \item $R_{\zero}(v,w)\, \wedge\, v \xto{u} t \implies
        \exists s.\, w \xto{u} s \, \wedge\, R_{\one}(t,s)$;
  \item $R_{\one}(t,s) \, \wedge \, t \to v \implies \exists w.\, s \To w \wedge R_{\zero}(v,w)$
  \item $R_{\one}(t,s) \, \wedge \, t \to t' \implies \exists s'.\, s \To s' \wedge R_{\one}(t',s')$.
  \end{enumerate}
\emph{Applicative similarity} $(\lesssim_{\zero},\lesssim_{\one})$ is the (sortwise) union of all applicative simulations.
\end{definition}
Similarly, \autoref{def:abslogrel} yields
\begin{definition}[Step-indexed logical relation for \skif] The \emph{step-indexed logical
relation} $\logrel = \bigcap_n \logrel^n$ is given by the relations $\logrel^{n} =
(\logrel^{n}_{\zero},\logrel^{n}_{\one})\pred{} \CLFT \times \CLFT$
defined inductively by $\logrel^{0}= \CLFT\times \CLFT$ and
\begin{align*}
\logrel^{n+1}_{\zero} &= \{\,(v,w) \mid \forall
  u,z.\,\logrel^{n}_{\zero}(u,z) \land v\xto{u}t
  \implies \exists s.\, w \xto{z} s \land \logrel^{n}_{\one}(t,s)\,\};\\
    \logrel^{n+1}_{\one} &= \{\,(t,s) \mid (t \to v
  \implies \exists w.\, s \To w \land \logrel_{\zero}^{n}(v,w)) \land (t \to t' \implies
  \exists s'.\, s \To s' \land \logrel_{\one}^{n}(t',s'))\,\}.
\end{align*} 
\end{definition}
As in the case of $\xCL$, since \skif is a deterministic language, it suffices to index over natural numbers instead of ordinals. As a consequence of the Abstract Soundness Theorem \ref{thm:congruence-ho-abstract-gsos}, we derive
\begin{theorem}[Soundness Theorem for \skif] Both $\appsim$ and $\logrel$ are sound for the contextual preorder: for all terms $p,q\in \CLFT$ of the same sort,
\[
p \appsim q \; \implies\; p\contpre q\qquad\text{and}\qquad
\logrel(p,q) \; \implies\; p\contpre q.\]
\end{theorem}

\begin{proof}
Clearly both $\appsim$ and $\logrel$ are contained in $O$ \eqref{eq:o-xclfg}. Therefore, we only need to check the conditions \ref{asm:1}--\ref{asm:6} of \autoref{thm:congruence-ho-abstract-gsos}. Condition \ref{asm:1} holds for all polynomial functors; note that strong epis and monos in $\Set^\twoc$ are the sortwise surjective and injective maps, resp. The order in \ref{asm:2} is given by pointwise equality in the value sort and pointwise inclusion in the computation sort; hence, due to the use of the left-to-right Egli-Milner lifting, \ref{asm:4} holds. \ref{asm:3} and \ref{asm:5} follow like for $\xCL$, since our lifting $\barB$ can be expressed as a canonical lifting of a deterministic behaviour bifunctor composed with the left-to-right lifting of the power set functor. The lax bialgebra condition \ref{asm:6} again means precisely that the operational rules remain sound when strong transitions are replaced by weak ones, which is easily verified by inspecting the rules. For illustration, consider the two rules for the application operator $\bullet$ and the weak version of the second rule:
\[
  \inference{t \to t'}{t\cccomp s\to t'\cccomp s}
    \qquad \inference{t\to v}{t\cccomp s\to v\vccomp s} \qquad
  \inference{t\To v}{t\cccomp s\To v\vccomp s}
\]
The weak rule is clearly sound, since it emerges via repeated application of the first rule, followed by a single application of the second rule. Similarly for the other rules of application. The rules of the remaining operators are premise-free, hence trivially sound w.r.t.\ weak transitions.
\end{proof}

\begin{example}[$\beta$-law]
  \label{ex:betalawfg}
  The $\beta$-law for \skif says that for all $v,w \in \CLFT_{\zero}$ and $t\in \CLFT_\one$,
  \[
    v \xto{w} t \implies v \vvcomp w \contpre_{\one} t \,\land\, t \contpre_{\one} v \vvcomp w.
  \]
  We make use of applicative similarity to prove the above statement. According
  to the soundness theorem, it suffices to come up with an applicative
  simulation which relates the two terms.
  \begin{enumerate}[label=(\roman*)]
  \item $v \vvcomp w \contpre_{\one} t$:\,\,\,Let $R = (\varnothing, \Delta \cup \{(v \vvcomp w, t)\})$. We have $v \vvcomp w \to t$ and $t
    \To t$ and $R_{\one}(t,t)$, which shows that $R$ is an applicative simulation.
  \item $t \contpre_{\one} v \vvcomp w$:\,\,\,  Let $R = (\Delta, \Delta\cup 
    \{(t,v \vvcomp w)\})$. If $t \to u$ for
    some value $u$, then $v \vvcomp w \To u$ and $R_{\zero}(u,u)$. Similarly, if $t\to t'$ for some computation $t'$, then $v \vvcomp w \To t'$ and $R_{\one}(t',t')$. Therefore $R$ is an applicative simulation.
  \end{enumerate}
\end{example}

We will next show the call-by-value $\eta$-law for \skif via the logical
relation $\logrel$. As an auxiliary result, let us first observe that $\logrel$ is closed
under backwards $\beta$-reduction:

\begin{lemma}
  \label{lem:backbeta}
  For all $t,t',s,s' \in \CLFT_{\one}$ and $n\in \Nat$ such that $\logrel^{n}_{\one}(t,s)$, we have
  \[
    \text{(i)}\,\,t' \to t \implies \logrel^{n}_{\one}(t',s) \qquad \text{and} \qquad
    \text{(ii)}\,\,s' \to s \implies \logrel^{n}_{\one}(t,s').
  \]
\end{lemma}

\begin{proof} Both statements are trivial for $n=0$. We prove them for a positive integer $n+1$:
  \begin{enumerate}[label=(\roman*)]
  \item Suppose $\logrel^{n+1}_{\one}(t,s)$ and $t' \to t$. Then $\logrel^{n+1}_{\one}(t',s)$ holds because $t' \to t$,
    $s \To s$ and $\logrel^{n}_{\one}(t,s)$. 
  \item Suppose $\logrel^{n+1}_{\one}(t,s)$ and $s' \to s$. We have to prove
    that $\logrel^{n+1}_{\one}(t,s')$. Suppose first that $t \to t'$ where $t' \in
    \CLFT_{\one}$. Then there exists $s''$ such that $s \To s''$ and
    $\logrel^{n}_{\one}(t',s'')$. Since $s'\To s''$, this shows $\logrel^{n+1}_{\one}(t,s')$. Analogous for the case $t\to v$ where $v\in \CLFT_{\zero}$.
 \qedhere
  \end{enumerate}
\end{proof}

\begin{example}[$\eta$-law]
  The (call-by-value) $\eta$-law in \skif says that, for all values $v \in
  \CLFT_{\zero}$,
  \begin{equation}
    \label{eq:exeta1}
    v \contpre_{\zero} S''([K'([I])], [v]) \,\land\, S''([K'([I])], [v]) \contpre_{\zero} v.
  \end{equation}
  Alternatively, another way to express the $\eta$-law is as
  \begin{equation}
    \label{eq:exeta2}
    [v] \contpre_{\one} (S \vccomp (K \vvcomp I)) \cvcomp v \,\land\,
    (S \vccomp (K \vvcomp I)) \cvcomp v \contpre_{\one} [v].
  \end{equation}
  We focus on \eqref{eq:exeta1}, as the proof of \eqref{eq:exeta2} follows a
  similar structure.
  By the soundness of $\L$, it suffices to prove that $\logrel^{n+1}_{\zero}(v, S''([K'([I])], [v]))$ and $\logrel^{n+1}_{\zero}(S''([K'([I])], [v]),v)$ for every $n\in \Nat$.
  \begin{enumerate}[label=(\roman*)]
  \item $\logrel^{n+1}_{\zero}(v, S''([K'([I])], [v]))$:\,\,\,
    It suffices to show that
    \[
      \forall {w,u}.\,\logrel^{n}_{\zero}(w,u) \land v \xto{w} t
      \implies \logrel_{\one}^{n}(t, ([K'([I])] \cvcomp u)\cccomp([v]\cvcomp u)).
    \]
    Suppose $v \xto{u} s$, which gives that $([K'([I])] \cvcomp
    u)\cccomp([v]\cvcomp u) \To s$. Since $\logrel$ is reflexive (being a congruence on an initial algebra), we have
    $\logrel_{\zero}(v,v)$ and thus $\logrel^{n}_{\one}(t,s)$.
    By \Cref{lem:backbeta}, we conclude $\logrel_{\one}^{n}(t, ([K'([I])]
    \cvcomp u)\cccomp([v]\cvcomp u))$.
  \item $\logrel^{n+1}_{\zero}(S''([K'([I])], [v]),v)$:\,\,\,
    This case is similar to the previous. It suffices to show that
    \[
      \forall {w,u}.\,\logrel^{n}_{\zero}(w,u) \land v \xto{u} s
      \implies \logrel_{\one}^{n}(([K'([I])] \cvcomp w)\cccomp([v]\cvcomp w),
      s).
    \]
    Suppose $v \xto{w}t$, thus $([K'([I])] \cvcomp w)\cccomp([v]\cvcomp w)
    \To t$. By reflexivity of $\logrel$, we get $\logrel_{\zero}(v,v)$ and so
    $\logrel^{n}_{\one}(t,s)$, which suffices by \Cref{lem:backbeta}.
  \end{enumerate}
\end{example}

\begin{remark}
  The results of our paper, which are themselves instances of a more
  general theory, are \emph{compositional} or \emph{modular} by nature. For
  example, we could enrich the \skif language by adding a unary fixpoint operator $\mathsf{fix}$
  to the grammar with the semantics:
  \[
    \inference{}{\mathsf{fix}(t) \to t\cvcomp S''(K \vvcomp I, \mathsf{fix}(t))}
  \] 
(mimicking the standard reduction $\mathsf{fix}\;f\to f\cdot (\lambda x.\, (\mathsf{fix}\, f)\, x)$ for the fixpoint combinator in call-by-value~\cite{Mitchell96}).
  Since the above rule is premise-free, it does not violate the lax bialgebra
  condition and hence applicative similarity and the step-indexed logical relation in the enriched \skif are still sound for the contextual preorder.
\end{remark}

\section{Call-by-push-value}
\label{sec:cbpv}

Levy's \emph{call-by-push-value} \cite{DBLP:phd/ethos/Levy01,
  DBLP:journals/siglog/Levy22} (CBPV) is a programming paradigm based on the
fundamental principle that computations are reducing program terms, while
values represent finished computations. Importantly, CBPV subsumes
both call-by-name and call-by-value semantics and also supports general
computational effects. We go on to implement a CBPV language 
with recursive types and nondeterministic choice, simply called \cbpv, in the
higher-order abstract GSOS framework. Subsequently, we develop a theory of
program equivalence in \cbpv as an application of the abstract results from \autoref{sec:abstract-op-methods}. 

The types of \cbpv are divided into \emph{value types} and
\emph{computation types}, and defined at the top of \Cref{fig:computations}.
There $\alpha$ ranges over a fixed countably infinite set of type variables. We denote by $\Tyv$ and $\Tyc$ the sets of closed value types and closed computation types, respectively, and by $\Ty=\Tyv+\Tyc$ the set of all closed types. We will use the metavariables $\tau, \tau_{1},\tau_{2},\dots$ to denote generic
types in $\Tyl$. The terms of \cbpv are intrinsically generated by the rules in
\Cref{fig:computations} (where $\Gamma$ ranges over contexts $\{x_1\colon \vtype_1,\ldots, x_n\colon \vtype_n\}$ of variables of value type) and the small-step (open
evaluation) operational semantics of \cbpv are presented in \Cref{fig:lamsemantics}. We annotate
operations with their type arguments (for example $\linl_{\vtype_{1},\vtype_{2}}(t) \c
\vtype_{1} \boxplus \vtype_{2}$) in places where these annotations are useful.
In addition, we will be writing $\Gamma \vdash t \c \tau$ if the sort of the
type $\tau$ is unspecified.

\begin{figure}
  \begin{flalign*}
		\makebox[0pt][l]{\qquad (Value types)}&&\hspace{10em}  \vtype_{1},\vtype_{2},\ldots&\Coloneqq U \ctype \mid \utype \mid \vtype_1\boxplus\vtype_2\mid\vtype_1\boxtimes\vtype_2,&\hspace{2cm}\\ 
		\makebox[0pt][l]{\qquad (Computation types)} && \hspace{10em}\ctype, \ctype_{1},\ctype_{2} \ldots&\Coloneqq \alpha \mid F \vtype \mid \arty{\vtype}{\ctype} \mid \ctype_{1} \otimes \ctype_{2} \mid
      \mu \alpha. \ctype&\\[-2ex]
	\end{flalign*}
\begin{flushleft}
Typing rules for values:\\[-2ex]
\end{flushleft}

\begin{gather*}
  \inference{\Gamma \vdash^{\zero} v \c \vtype_{1}}{\Gamma \vdash^{\zero} \linl_{\vtype_{1},\vtype_{2}}(v) \c \vtype_{1}
    \boxplus \vtype_{2}}
  \qquad
  \inference{\Gamma \vdash^{\zero} v \c \vtype_{2}}{\Gamma \vdash^{\zero} \linr_{\vtype_{1},\vtype_{2}}(v) \c \vtype_{1}
    \boxplus \vtype_{2}}
  \qquad
  \inference{x \c \vtype \in \Gamma}{\Gamma \vdash^{\zero} x \c \vtype}
  \\[1ex]
  \inference{}{\Gamma \vdash^{\zero} \star \c \utype}
  \qquad
  \inference{\Gamma \vdash^{\one} t \c \ctype}{\Gamma \vdash^{\zero} \mathsf{thunk}_{\ctype}(t) \c
    U \ctype}
  \qquad
  \inference{\Gamma \vdash^{\zero} v \c \vtype_{1} &
    \Gamma \vdash^{\zero} w \c \vtype_{2}}{\Gamma \vdash^{\zero}  \mathsf{pair}_{\vtype_{1},\vtype_{2}}(v,w)
  \c \vtype_{1} \boxtimes \vtype_{2}}
\end{gather*}

\bigskip
\begin{flushleft}
Typing rules for computations:\\[-2ex]
\end{flushleft} 

\begin{gather*}
  \inference{\Gamma \vdash^{\zero} v \c \vtype}{\Gamma \vdash^{\one} \mathsf{prod}_{\vtype}(v)
    \c F\vtype} \qquad
  \inference{\Gamma \vdash^{\zero} v \c U \ctype}{\Gamma \vdash^{\one} \mathsf{force}_{\ctype}(v)
    \c \ctype} \qquad
  \inference{\Gamma \vdash^{\zero} v \c \vtype & \Gamma \vdash^{\one} t \c
    \arty{\vtype}{\ctype}}{\Gamma \vdash^{\one} \mathsf{app}_{\vtype,\ctype}(t,v) \c \ctype}
  \\[1ex]
  \inference{\Gamma \vdash^{\one} t \c \ctype[\mu \alpha.\ctype/\alpha]}{\Gamma \vdash^{\one} \mathsf{fold}_{\ctype}(t)
    \c \mu \alpha.\ctype} \qquad
  \inference{\Gamma \vdash^{\one} t \c \mu \alpha.\ctype}{\Gamma \vdash^{\one} \mathsf{unfold}_{\ctype}(t)
    \c \ctype[\mu \alpha.\ctype/\alpha]}
  \\[1ex]
  \inference{\Gamma, x \c \vtype \vdash^{\one} t \c \ctype}
  {\Gamma \vdash^{\one} \mathsf{lam}_{\ctype}\,x\c\vtype.\, t\c \arty{\vtype}{\ctype}}
  \qquad
  \inference{\Gamma \vdash^{\one} t \c \ctype & \Gamma \vdash^{\one} s \c \ctype}
  {\Gamma \vdash^{\one} t \oplus_{\ctype} s \c \ctype}
  \qquad
  \inference{\Gamma \vdash^{\one} s \c F \vtype & \Gamma,x \c \vtype \vdash^{\one} t \c
    \ctype}{\Gamma \vdash^{\one} \toexprt{s}{x}{t}{\vtype}{\ctype} \c \ctype}
  \\[1ex]
  \inference{\Gamma \vdash^{\zero} v \c \vtype_{1} \boxplus \vtype_{2}
    & \Gamma,x \c \vtype_{1} \vdash^{\one} s \c \ctype
    & \Gamma,y \c \vtype_{2} \vdash^{\one} r \c \ctype}
  {\Gamma \vdash^{\one} \mathsf{case}_{\vtype_{1},\vtype_{2},\ctype}(v,x.s,y.r) \c \ctype}
  \qquad
  \inference{\Gamma \vdash^{\one} t \c \ctype_{1} \otimes \ctype_{2}}
  {\Gamma \vdash^{\one} \lfst_{\ctype_{1},\ctype_{2}}(t) \c \ctype_{1}}  \\[1ex]
  \inference{\Gamma \vdash^{\one} t \c \ctype_{1} \otimes \ctype_{2}}
  {\Gamma \vdash^{\one} \lsnd_{\ctype_{1},\ctype_{2}}(t) \c \ctype_{2}}
  \qquad
  \inference{\Gamma \vdash^{\one} t \c \ctype_{1} & \Gamma \vdash^{\one} s \c \ctype_{2}}
  {\Gamma \vdash \mathsf{pair}_{\ctype_{1},\ctype_{2}}(t,s) \c \ctype_{1} \otimes \ctype_{2}}
  \qquad
  \inference{\Gamma \vdash^{\zero} v \c \vtype_{1} \boxtimes \vtype_{2}
    & \Gamma,x \c \vtype_{1},y \c \vtype_{2} \vdash^{\one} t \c \ctype}
  {\Gamma \vdash^{\one} \mathsf{pm}_{\vtype_{1},\vtype_{2},\ctype}(v,(x,y).t) \c \ctype}
\end{gather*}
\caption{Syntax of \cbpv.}
\label{fig:computations}
\end{figure}

We briefly explain the core ideas behind \cbpv. The overarching principle is
that computations ``compute'' (i.e. $\beta$-reduce) by manipulating values. The
value operations
in \Cref{fig:computations} are mostly straightforward, apart from the
$\mathsf{thunk}$ expression: placing a computation inside a $\mathsf{thunk}$
expression essentially ``freezes'' the computation and turns it into a value.
Such a frozen computation may be resumed via the $\mathsf{force}$ expression.
Conversely, the $\mathsf{prod}$ expression turns any value into a computation; a
term $\mathsf{prod}(v)$ represents a finished computation, with the outcome
being $v$. The operation
$\mathsf{to}-\mathsf{in}$ is useful for sequencing computations: a term
$\toexpr{s}{x}{t}$ first evaluates $s$ until it produces a value $v$; it
then moves to $t$, making sure to utilise the produced value $v$.

\subsection{Modelling \cbpv in Higher-Order Abstract GSOS}

The categorical modelling of \cbpv takes place in a universe of (covariant)
presheaves, namely $(\Set^{\fset/{\Tyv}})^{\Tyl}$, based on the established theory of
abstract syntax and variable binding~\cite{DBLP:conf/lics/FiorePT99,DBLP:conf/csl/FioreH10,DBLP:journals/mscs/Fiore22},
and following up on earlier work on modelling call-by-name $\lambda$-calculi in higher-order abstract
GSOS~\cite{gmstu24,gmtu24lics}. One important difference between the present
and the aforementioned works is that the categorical notions of typing contexts,
variables and substitution need to be carefully adjusted to apply to a two-sorted setting with values and computations, in which
only values can be substituted for variables.

\begin{figure}
  \begin{gather*}
  \inference{}{\mathsf{fold}(\mathsf{unfold}(t)) \to t}
    \qquad
  \inference{t\to t'}{\mathsf{app}(t,v)\to \mathsf{app}(t',v)}
  \qquad
  \inference{}{\mathsf{app}((\mathsf{lam}\,x\c\vtype.\,t),v) \to t[v/x]}
  \\[1ex]
  \inference{t \to t'}{\lfst(t) \to \lfst(t')}
  \qquad
  \inference{t \to t'}{\lsnd(t) \to \lsnd(t')}
  \qquad
  \inference{}{\lfst(\mathsf{pair}(t,s)) \to t}
  \qquad
  \inference{}{\lsnd(\mathsf{pair}(t,s)) \to s}
  \\[1ex]
  \inference{}{t \oplus s \to t}
  \qquad
  \inference{}{t \oplus s \to s}
  \qquad
  \inference{s \to s'}{\toexpr{s}{x}{t} \to \toexpr{s'}{x}{t}}
  \qquad
  \inference{s \to \mathsf{prod}(v)}{\toexpr{s}{x}{t} \to \mathsf{app}((\mathsf{lam}\,x.t),v)}
  \\[1ex]
  \inference{}{\mathsf{case}(\linl(v),x.s,y.r)\to
    \mathsf{app}((\mathsf{lam}\,x.s),v)} \qquad
  \inference{}{\mathsf{case}(\linr(v),x.s,y.r)\to
    \mathsf{app}((\mathsf{lam}\,y.r),v)}
  \\[1ex]
  \inference{}{\mathsf{force}(\mathsf{thunk}(t)) \to t}
  \qquad
  \inference{}{\mathsf{pm}(\mathsf{pair}(v,w),(x,y).t) \to
    \mathsf{app}((\mathsf{lam}\,y.\mathsf{app}((\mathsf{lam}\,x.t),v)),w)}
\end{gather*}
\caption{Small-step operational semantics for \cbpv.}
\label{fig:lamsemantics}
\end{figure}

\paragraph{The category of typing contexts.}

Let $\fset/{\Tyv}$ be the free cocartesian category over the set of
(value) types $\Tyv$ or, equivalently, the comma category $\fset \xto{J} \Set
\xleftarrow{\Tyv} 1$, where $J$ is the inclusion of the category $\fset$ of
finite cardinals to the category of sets. In detail, the objects of
$\fset/{\Tyv}$ are functions
$\Gamma \c n=\{1,\dots,n\} \to \Tyv$, while morphisms $r \c \Gamma \to
\Delta$ are functions $|r| \c \dom(\Gamma) \to
\dom(\Delta)$ such that $\Gamma = \Delta \comp |r|$. The identity
morphism on an object $\Gamma$ is given by the identity function $\id \c
\dom(\Gamma) \to \dom(\Gamma)$ and composition in $\fset/{\Tyv}$
is given by function composition.

The
category $\fset/{\Tyv}$ is intuitively understood as the category of
\emph{typing contexts}: an object $\Gamma \c n \to \Tyv$ represents the context
$\{x_{1} \c \Gamma(1), x_{2} \c \Gamma(2),\dots,x_{n} \c \Gamma(n)\}$.
Morphisms of $\fset/{\Tyv}$ are \emph{renamings}: a morphism $r \c
\Gamma \to \Delta$ from $\Gamma$ to $\Delta$ maps each variable
$x_{i} \in \Gamma$ to the corresponding variable $x_{|r|(i)} \in \Delta$. The
condition $\Gamma = \Delta \comp |r|$ translates to $x_{i}$ and
$x_{|r|(i)}$ being of the same type. Each value type $\vtype \in \Tyv$ induces the
single-variable context 
  $\check \vtype \c 1 \to \Tyv$, such that $\check \vtype(\star) = \vtype$
  or informally,
  $\check\vtype = \{x_{1} \c \vtype\}$.
The initial object of\/ $\fset/{\Tyv}$ is the empty context $\varnothing \c 0 \to
\Tyv$; coproducts in $\fset/{\Tyv}$ are formed by copairing:
\begin{equation}
  \label{eq:coprod}
  (\Gamma \c \dom(\Gamma) \to \Tyv) + (\Delta
  \c \dom(\Delta)
  \to \Tyv)
  = [\Gamma, \Delta] \c \mathrm{dom}(\Gamma) + \mathrm{dom}(\Delta) \to \Tyv.
\end{equation}

\begin{notation}
  We introduce the following notation when $\Delta$ in \eqref{eq:coprod}
  is $\check\vtype$ for some type $\vtype \in \Tyv$:
  \begin{equation}
  \label{coproduct}
  \Gamma \xrightarrow{\oname{old}^{\vtype}_{\Gamma}} \Gamma +
  \check\vtype \xleftarrow{\oname{new}^{\vtype}_{\Gamma}} \check\vtype
\end{equation}
Here, $\oname{old}^{\vtype}_{\Gamma}$ maps a variable $x_{i}\in \Gamma$ to
the variable $x_{i}\in \Gamma + \check\vtype$ and the variable $x \in
\check\vtype$ to $x_{|\Gamma|+1}$.
\end{notation}

\paragraph{The category of variable sets.}
Moving over to category $(\Set^{\fset/{\Tyv}})^{\Tyl}$, its objects are
$\Tyl$-indexed families $(X_{\tau})_{\tau \in \Tyl}$ of \emph{presheaves} in
$\Set^{\fset/{\Tyv}}$. Given $X\in (\Set^{\fset/{\Tyv}})^{\Tyl}$ we write $X_{\tau}(\Gamma)$ for $X(\tau)(\Gamma)$ and $X_\tau(r)$ for $X(\tau)(r)$. Our main example of an object in
$(\Set^{\fset/{\Tyv}})^{\Tyl}$ is the presheaf $\Lambda$ of \cbpv-terms:
\begin{equation}
  \label{eq:presheaflambda}
  \Lambda_{\tau}(\Gamma) = \{t \mid \Gamma \vdash t \c \tau \},
\end{equation}
where all terms are taken modulo $\alpha$-equivalence. The map $\Lambda_{\tau}(r)$ acts by renaming free variables in terms according to $r\colon \Gamma\to \Delta$.

A morphism (i.e., natural transformation) $f \c X \to Y$ in $(\Set^{\fset/{\Tyv}})^{\Tyl}$ is a
family of functions $f_{\tau,\Gamma}\colon X_{\tau}(\Gamma)\to Y_{\tau}(\Gamma)$,
indexed by $\tau \in \Tyl$ and $\Gamma \in \fset/{\Tyv}$, that is
compatible with renaming:
\[
  \forall r \c \Gamma \to \Delta.\, Y_{\tau}(r) \comp f_{\tau,\Gamma} =
  f_{\tau,\Delta} \comp X_{\tau}(r).
\]
A \emph{relation} on $X\in (\Set^{\fset/{\Tyv}})^{\Tyl}$ is represented by a monomorphism $R
\monoto X \times X$, that is, a family of relations $R_{\tau}(\Gamma)
\subseteq X_{\tau}(\Gamma) \times X_{\tau}(\Gamma)$ that are compatible with
renamings. For any $X \in (\Set^{\fset/{\Tyv}})^{\Tyl}$, relations on $X$ form a
complete lattice in the expected way: the join $\bigcup_i R_i$ of relations $R_i \monoto X \times X$ ($i\in I$) is given by $(\bigcup_i R_i)_{\tau}(\Gamma)=\bigcup_i (R_i)_{\tau}(\Gamma)$ for
each $\tau$ and $\Gamma$.

\begin{notation}
  We occasionally omit the subscripts $\tau$ and $\Gamma$ when applying
  morphisms in $(\Set^{\fset/{\Tyv}})^{\Tyl}$, e.g., we write $f(p)$ instead of
  $f_{\tau,\Gamma}(p)$ for $f\colon X\to Y$ and $p \in X_{\tau}(\Gamma)$.
\end{notation}

\paragraph{Core constructions in $(\Set^{\fset/{\Tyv}})^{\Tyl}$.}

We set up the basic definitions in $(\Set^{\fset/{\Tyv}})^{\Tyl}$ that enable
the categorical modelling of the syntax and semantics of call-by-value languages
with variable binding and substitution. The presheaf $V\in (\Set^{\fset/{\Tyv}})^{\Tyl}$ of \emph{variables} is given by $V_\kappa(\Gamma)=\emptyset$ and
\[
  V_{\vtype}(\Gamma) = \{x \in \dom(\Gamma) \mid \Gamma(x) = \vtype\},
  \qquad
  V_{\vtype}(r)(x) = r(x) \,\,\text{ for } r \c
  \Gamma \to \Delta,
\]

An important construction in $(\Set^{\fset/{\Tyv}})^{\Tyl}$ is the
\emph{substitution tensor} $- \bullet -$. It is a novel
variation of the substitution tensor of \cite{DBLP:conf/lics/FiorePT99}, and is defined as the following
coend for $X,Y\in (\Set^{\fset/{\Tyv}})^{\Tyl}$:  
\begin{equation}
  (X \bullet_{\tau} Y) (\Gamma) = \int^{\Delta \in {\fset/{\Tyv}}} X_{\tau}(\Delta)
  \times \prod\nolimits_{i \in \dom(\Delta)}Y_{\Delta(i)}(\Gamma).
\end{equation}

More explicitly, elements of $(X \bullet_{\tau} Y) (\Gamma)$ are terms paired with
\emph{substitutions}, meaning triples $(\Delta,t,\vec{u})$ where $\Delta
\in \fset/{\Tyv}$,
$t \in X_{\tau}(\Delta)$ and $\vec{u} \in \prod_{i \in
  \dom(\Delta)}Y_{\Delta(i)}(\Gamma)$, modulo the equivalence relation generated by
$(\Delta,t,\vec{u}) \approx (\Delta',t',\vec{u}')$, identifying such
triples if and only if there exists $r \c \Delta \to \Delta'$ with
$X_{\tau}(r)(t) = t'$ and $u(i) = u'(r(i))$\,\footnote{In the published version of the current
  manuscript~\cite{10.1145/3704871}, we incorrectly state that
  $(\Set^{\fset/{\Tyv}})^{\Tyl}$ is monoidal w.r.t $\bullet$ and $V$. We wish to thank
  Ohad Kammar for pointing this out.}.
For each $Y \in (\Set^{\fset/{\Tyv}})^{\Tyl}$, the functor $(\argument)
\bullet Y \c (\Set^{\fset/{\Tyv}})^{\Tyl} \to (\Set^{\fset/{\Tyv}})^{\Tyl}$ has
a right adjoint $\llangle Y, \argument \rrangle$,
given by
\begin{equation}
  \llangle Y,Z \rrangle_{\tau}(\Gamma) = \Set^{\fset/{\Tyv}}\bigl(\prod\nolimits_{i \in
    \dom(\Gamma)}Y_{\Gamma(i)}, Z_{\tau}\bigr).
\end{equation}
An element of $\llangle Y,Z \rrangle_{\tau}(\Gamma)$, namely a natural transformation $f\colon \prod_{i\in \dom(\Gamma)} Y_{\Gamma(i)}\to Z$, models the \emph{simultaneous
substitution} of a tuple of $Y$-terms of length $\dom(\Gamma)$ that
produces a term in $Z_{\tau}$.

\begin{notation}
  We use the following notation for the natural isomorphism witnessing the
  adjoint situation $(\argument) \bullet Y\dashv \llangle Y, \argument \rrangle$:\\[-2ex]
  \[
    \begin{tikzcd}
      (\Set^{\fset/{\Tyv}})^{\Tyl}(X \bullet Y, Z)
      \ar[phantom]{r}{\cong}
      \ar[bend left=1em]{r}{(\argument)^{\flat}}
      & (\Set^{\fset/{\Tyv}})^{\Tyl}(X, \llangle Y , Z \rrangle) \ar[bend
      left=1em]{l}{(\argument)^{\sharp}}
    \end{tikzcd}
  \]
\end{notation}
As an example of simultaneous substitution, there is map $\mathrm{sub} \c
\Lambda \bullet \Lambda \to \Lambda$ (equivalently, a map
$\mathrm{sub}^{\flat} \c \Lambda \to \llangle \Lambda,\Lambda\rrangle$),
mapping pairs of terms $\Delta \vdash t \c \tau$ and substitutions
$\vec{u} \c \prod_{i \in \dom(\Delta)}\Lambda_{\Delta(i)}(\Gamma)$ to $\Gamma \vdash t[\vec{u}] \c \tau$, where $t[\vec{u}]$ denotes the term obtained by substituting $\vec{u}$ for the free variables in $t$. Since morphisms in
$(\Set^{\fset/{\Tyv}})^{\Tyl}$ are natural transformations, naturality for
$\mathrm{sub} \c \Lambda \bullet \Lambda \to \Lambda$ means that substitution in
\cbpv commutes with renamings of free variables. 

We next look at exponentials in $\Set^{\fset/{\Tyv}}$. Given $X,Y\in \Set^{\fset/\Tyv}$, the exponential $Y^{X}$ and
its evaluation morphism \mbox{$\ev\c Y^X \times X \to Y$} in $\Set^{\fset/{\Tyv}}$
are respectively computed by the following formulas, standard in presheaf
toposes~\cite[Sec.~I.6]{MacLaneMoerdijk92}:
\[
  Y^{X}(\Gamma) =  \mbox{$\Set$}^{\fset/{\Tyv}}\bigl({\fset/{\Tyv}}(\Gamma,
  \argument) \times X, Y\bigr) \quad \text{and} \quad
  \ev_{\Gamma}(f, x) = f_{\Gamma}(\id_{\Gamma}, x) \in Y(\Gamma).
\]
We write $f(x)$ for $\ev_{\Gamma}(f,x)$. This allows us to construct presheaves such as
$Y_{\ctype}^{X_{\vtype}}$ to represent a space of \emph{functions} of the
type $\arty{\vtype}{\ctype}$.

To each value type $\vtype \in \Tyv$ we associate the
endofunctor $\delta^{\vtype} \c \Set^{\fset/{\Tyv}} \to \Set^{\fset/{\Tyv}}$,
where
\[
  \delta^{\vtype}X(\Gamma) = X(\Gamma + \check\vtype).
\]
The endofunctor $\delta^{\vtype}$ abstractly captures \emph{variable binding}: informally, an element of $\delta^{\vtype} X(\Gamma)$ arises from a $X$-term in context $\Gamma + \check\vtype$ by binding the distinguished variable of type $\vtype$. 

One useful fundamental operation in $\Set^{\fset/{\Tyv}}$ is that of
\emph{weakening} given by the natural transformation $\up^{\vtype} \c \mathrm{Id} \to \delta^{\vtype}$, which
is induced by the coproduct structure of $\fset/{\Tyv}$:
\begin{equation}\label{eq:up}
  \up_{X,\Gamma}^{\vtype} = X(\oname{old}_{\Gamma}^{\vtype})\c X(\Gamma)\to X(\Gamma+\check\vtype).
\end{equation}
Informally, this operation regards a term in context $\Gamma$ as a term in context $\Gamma + \check\vtype$.

\paragraph{Categorical modelling of \cbpv} We are prepared to implement the
syntax and semantics of \cbpv in the higher-order abstract GSOS framework. The syntax of the language is modeled by the syntax endofunctor $\Sigma \c
(\Set^{\fset/{\Tyv}})^{\Tyl} \to (\Set^{\fset/{\Tyv}})^{\Tyl}$ that is canonically
induced by the ({binding}) signature of \cbpv. It is defined as follows:
\begin{equation}
  \label{eq:cbpvsyn}
  \begin{aligned}
    & \Sigma_{\tau} X = \underbrace{V_{\tau}}_{\text{Variables}} +\, \Sigma_{\tau}^{1}X + \Sigma_{\tau}^{2}X + \Sigma^{3}_{\tau}X + \underbrace{\coprod\nolimits_{\ctype \c \tau = \ctype[\mu \alpha.\ctype/\alpha]}X_{\mu \alpha.\ctype}}_{\mathsf{unfold}}, \quad \text{where}
    \\
    & \Sigma^{1}_{\vtype_{1} \boxplus \vtype_{2}} X = \underbrace{X_{\vtype_{1}}}_{\linl} + \underbrace{X_{\vtype_{2}}}_{\linr}, \quad
      \Sigma^{1}_{\vtype_{1} \boxtimes \vtype_{2}} X = \underbrace{X_{\vtype_{1}} \times X_{\vtype_{2}}}_{\mathsf{pair}_{\vtype_{1},\vtype_{2}}}, \quad
      \Sigma^{1}_{\utype} X = \underbrace{1}_{\star},
      \quad \Sigma^{1}_{U\ctype}X = \underbrace{X_{\ctype}}_{\mathsf{thunk}},
    \\
    & \Sigma^{1}_{F\vtype} X = \underbrace{X_{\vtype}}_{\mathsf{prod}}, \qquad
      \Sigma^{1}_{\arty{\vtype}{\ctype}} X = \underbrace{\delta^{\vtype}X_{\ctype}}_{\mathsf{lam}}, \qquad
      \Sigma^{1}_{\ctype_{1} \otimes \ctype_{2}} X  =
      \underbrace{X_{\ctype_{1}} \times X_{\ctype_{2}}}_{\mathsf{pair}_{\ctype_{1},\ctype_{2}}},
      \qquad \Sigma^{1}_{\mu \alpha.\ctype} X  =
      \underbrace{X_{\ctype[\mu \alpha.\ctype/\alpha]}}_{\mathsf{fold}},
    \\
    & \Sigma^{2}_{\ctype} X = \coprod_{\vtype_{1} \in \Tyv}
      \coprod_{\vtype_{2} \in \Tyv}
      \underbrace{X_{\vtype_{1} \boxplus \vtype_{2}}
      \times \delta^{\vtype_{1}}X_{\ctype}
      \times \delta^{\vtype_{2}}X_{\ctype}}_{\mathsf{case}}
      + \underbrace{X_{\vtype_{1} \otimes \vtype_{2}} \times
      \delta^{\vtype_{2}}\delta^{\vtype_{1}}X_{\ctype}}_{\mathsf{pm}},
      \quad \Sigma^{2}_{\vtype}X = 0,
    \\
    & \Sigma^{3}_{\ctype} X =
      \coprod_{\vtype \in \Tyv}
      (\underbrace{X_{\arty{\vtype}{\ctype}} \times X_{\vtype}}_{\mathsf{app}}
      + \underbrace{X_{F\vtype} \times \delta^{\vtype}X_{\ctype}}_{\mathsf{to}\text{-expression}}) +
      \coprod_{\ctype_{1} \in \Tyc}(\underbrace{X_{\ctype \boxtimes \ctype_{1}}}_{\lfst}
      + \underbrace{X_{\ctype_{1} \boxtimes \ctype}}_{\lsnd}) +
      \underbrace{X_{U\ctype}}_{\mathsf{force}} +
      \underbrace{X_{\ctype} \times X_{\ctype}}_{\argument\oplus\argument},
      ~~~
      \Sigma^{3}_{\vtype}X = 0.
  \end{aligned}
\end{equation}
The initial $\Sigma$-algebra
$\mu\Sigma$ is carried by the presheaf $\Lambda$
of \cbpv-terms\,\eqref{eq:presheaflambda}, cf.~\cite{DBLP:conf/csl/FioreH10}.

The behaviour bifunctor $B\colon (\cbpvcat)^\opp\times (\Set^{\fset/{\Tyv}})^{\Tyl}\to (\Set^{\fset/{\Tyv}})^{\Tyl}$ for \cbpv is given by
\begin{equation} 
  \label{eq:cbpvbeh}
  B(X,Y)= \llangle X,Y \rrangle \times \Pow_{\star}B'(X,Y),
\end{equation}
where $(\Pow_{\star}X)_{\tau} = \mathcal{\Pow} \cdot X_{\tau}$ (recall that
$\Pow \c \Set \to \Set$ is the power set functor) and
\begin{equation*}
    \begin{aligned}
      B'_{\vtype}(X,Y)
      &\;= D_{\vtype}(X,Y),
      &B'_{\ctype}(X,Y)
      &\;= Y_{\ctype} + D_{\ctype}(X,Y) \\
      D_{U\ctype}(X,Y)
      &\;= Y_{\ctype}
      &D_{\mu \alpha.\ctype}(X,Y)
      &\;= Y_{\ctype[\mu \alpha.\ctype/\alpha]}
      \\
      D_{\utype}(X,Y)
      &\,= 1,
      & D_{\vtype_{1} \boxplus \vtype_{2}}(X,Y)
      &\,= Y_{\vtype_{1}} +  Y_{\vtype_{2}}
      & D_{\vtype_{1} \boxtimes \vtype_{2}}(X,Y)
      &\;= Y_{\vtype_1} \times Y_{\vtype_2}, \\
      D_{\arty{\vtype}{\ctype}}(X,Y)
    &\;= Y_{\ctype}^{X_{\vtype}},
    & D_{F\vtype}(X,Y)
      &\;= Y_{\vtype},
    & D_{\ctype_{1} \otimes \ctype_{2}}(X,Y)
    &\;= Y_{\ctype_1} \times Y_{\ctype_2}.
  \end{aligned}
\end{equation*}
The component
$B'_{\tau}(X,Y)$ gives the range of dynamic behaviour of terms of type $\tau$.
For example, a computation $\Gamma \vdash t \c
\arty{\vtype}{\ctype}$ may $\beta$-reduce, witnessed by the left component $Y_{\arty{\vtype}{\ctype}}$ in
$B'_{\arty{\vtype}{\ctype}}(X,Y)$, or act as a function on terms, witnessed by 
$D_{\arty{\vtype}{\ctype}}(X,Y)=Y_\ctype^{X_\vtype}$.
The full behaviour bifunctor $B$ is the cartesian
product of $\llangle X,Y \rrangle $ and $\Pow_{\star}B'(X,Y)$. Hence, we consider each term as exhibiting two types of behaviour,
the first being its simultaneous substitution structure and the second being its
dynamics. The operational model of the higher-order GSOS law for
\cbpv should thus be given by
\begin{equation}\label{eq:op-model-cbpv}
  \langle \mathrm{sub}^\flat, \gamma^{2} \rangle \c \Lambda \to
  \llangle \Lambda,\Lambda \rrangle \times \Pow_{\star}B'(\Lambda,\Lambda),
\end{equation}
where $\mathrm{sub}^\flat$ and $\gamma^{2}$ model, resp., simultaneous substitution and the
dynamics of \cbpv. 

We capture the operational semantics of \cbpv by a $V$-pointed
higher-order GSOS law $\rho$ of $\Sigma$ \eqref{eq:cbpvsyn} over $B$
\eqref{eq:cbpvbeh}. The component
\[\rho_{(X,\var),Y}\colon \Sigma(X \times \llangle X,Y \rrangle
    \times \Pow_{\star}B'(X,Y)) \to
    \llangle X,\Sigmas(X+Y)\rrangle \times \Pow_{\star}B'(X,\Sigma^{\star}(X + Y)) \] 
is defined as a pairing $\langle\rho^{1}_{(X,\var),Y} \comp
\Sigma p,\,\rho^{2}_{(X,\var),Y}\rangle$, where $p$ is the middle product projection and
\begin{equation}
  \label{eq:rhorho}
    \begin{aligned}
  & \rho^{1}_{(X,\var),Y} \c \Sigma\llangle X,Y \rrangle \to
    \llangle X, \Sigma^{\star}(X + Y) \rrangle, \\
  & \rho^{2}_{(X,\var),Y} \c \Sigma(X \times \llangle X,Y \rrangle
    \times \Pow_{\star}B'(X,Y)) \to
    \Pow_{\star}B'(X,\Sigma^{\star}(X + Y)).
  \end{aligned}
\end{equation}
Here $\rho^{1}$ models the simultaneous substitution structure
of $\Lambda$ in the sense that it induces the map $\mathrm{sub}^{\flat} \c
\Lambda \to \llangle \Lambda,\Lambda\rrangle$. For each $\tau\in\Ty$ and $\Gamma\in \fset/\Tyv$, the component
\[ \rho^{1}_{(X,\var),Y,\tau,\Gamma} \c \Sigma_\tau\llangle X,Y \rrangle(\Gamma) \to
    \Set^{\fset/{\Tyv}}\bigl(\prod\nolimits_{i \in
    \dom(\Gamma)}X_{\Gamma(i)}, \Sigmas_{\tau}(X+Y)\bigr) \]
is defined as follows. (We only give a few exemplary clauses; a full definition of $\rho^{1}_{(X,\var),Y,\tau,\Gamma}$ can be found in the extended arXiv version~\ifarx{\autoref{app:rho-cbpv}}{\cite[App.~A]{gtu25_arxiv}} of our paper.)
\begin{alignat*}{4}
& x,\vec{u}
    && \quad \mapsto \quad && \vec{u}_{x} &&\qquad (x\in V_\tau(\Gamma), \text{ i.e. } \Gamma(x)=\tau \})\\
& \mathsf{app}_{\vtype,\ctype}(f,g),\vec{u}
    &&\quad\mapsto\quad && \mathsf{app}_{\vtype,\ctype}(f_\Delta(\vec{u}),g_\Delta(\vec{u}))  && \hspace*{-2.6em}(f\in \llangle X,Y\rrangle_{\arty{\vtype}{\ctype}}(\Gamma),\, g\in \llangle X,Y\rrangle_{\vtype}(\Gamma))\\
& \mathsf{force}_\ctype(f),\vec{u}
    &&\quad\mapsto\quad && \mathsf{force}_\ctype(f_\Delta(\vec{u})) &&\qquad\qquad \hspace*{1.0em} (f\in \llangle X,Y\rrangle_{U\ctype}(\Gamma))\\
& \mathsf{lam}_\ctype\,x\c\vtype. f,\vec{u}
    &&\quad\mapsto\quad && \mathsf{lam}_\ctype\,x\c\vtype.f_{\Delta+\check\vtype}(\vec{\up}_{X,\Delta}^\vtype(\vec{u}), \var_{\Delta+\check\vtype,\vtype}(\oname{new}_{\Delta}^{\vtype})) && \qquad\qquad\hspace*{1em} (f\in \delta^{\vtype}\llangle X,Y\rrangle(\Gamma))
\end{alignat*}
This requires some explanation. We fix a context $\Delta\in \fset/\Phi$ and use ``curried'' notation; for
instance, the third clause states that
$\rho^{1}_{\vtype_{1} \boxplus \vtype_{2},\Gamma}(\mathsf{force}_\ctype(f))
 = \lambda \vec{u}.\mathsf{force}_\kappa(f_\Delta(\vec{u}))$), where $\vec{u}\in \prod_{i\in \dom_{\Gamma(i)}} X_{\Gamma(i)}(\Delta)$. In the first clause, $\vec{u}_i\in X_{\Gamma(i)}(\Delta)$ denotes the $i$-th entry of $\vec{u}$.
In the last clause, we use the operation $\vec{\up}^{\vtype}$ that applies the operation $\up^{\vtype}$ \eqref{eq:up} componentwise, i.e.\  it takes $\vec{u}$ to $\vec{\up}^{\vtype}(\vec{u})
\in \prod_{i \in
\dom(\Gamma)}X_{\Gamma(i)}(\Delta + \check\vtype)$. Unfolding the last clause takes some care (see~\ifarx{\autoref{app:rho-cbpv}}{\cite[App.~A]{gtu25_arxiv}} for details); intuitively, this clause simply expresses that a simultaneous substitution $(\lambda x.t)[\vec{u}]$ is given by $\lambda x.t[\vec{u},x]$.

  Similarly, $\rho^{2}$ models the dynamics of \cbpv. Its component at $\tau\in \Ty$ and $\Gamma\in\fset/\Tyv$ is the map
\[    \rho^{2}_{(X,\var),Y,\tau,\Gamma} \c
     \Sigma_\tau(X \times \llangle X,Y \rrangle
      \times \Pow_{\star}B'(X,Y))(\Gamma)
     \to \Pow\cdot B'_\tau(X,\Sigma^{\star}(X + Y))(\Gamma);
\]
defined as follows. (Again, we give a few exemplary clauses and refer to~\ifarx{\autoref{app:rho-cbpv}}{\cite[App.~A]{gtu25_arxiv}} for a full definition.)
{\allowdisplaybreaks
  \begin{alignat*}{4}
    &{x}
    && \quad\mapsto\quad && \emptyset \\[-1ex]
    &{\mathsf{app}_{\vtype,\ctype}((-,-,L),(v,-,-))}
    && \quad\mapsto\quad && \{ f(v) \mid f \in L \cap
       Y_{\ctype}^{X_{\vtype}}(\Gamma)\}\,{\cup}\, \{\mathsf{app}_{\vtype,\ctype}(t', v) \mid t' \in L \cap
        Y_{\arty{\vtype}{\ctype}}(\Gamma)\} \\
    &{\mathsf{prod}_{\vtype}(v,-,-)}
    && \quad\mapsto\quad && \{v\}\\
    &{\mathsf{lam}_{\ctype}\,x\c
      \vtype.\,(-,f,-)}
    && \quad\mapsto\quad && \{\lambda e. f(x_{1},\dots,x_{\dom(\Gamma)},e)\}
\end{alignat*}
Here $x\in V_\vtype(\Gamma)$ in the first clause, $L\in \Pow\cdot B'_{\arty{\vtype}{\ctype}}(X,Y)(\Gamma)$ and $v\in X_{\vtype}(\Gamma)$ in the second clause, $v\in X_\vtype$ in the third clause, and $f\in \delta^\vtype\llangle X,Y \rrangle_{\ctype}(\Gamma)$ in the last clause. Moreover, the expression $\lambda e. f(x_{1},\dots,x_{\dom(\Gamma)},e)$ on the right-hand side denotes the element of $(\Sigmas_\ctype(X+Y))^{X_\vtype}(\Gamma) = \mbox{$\Set$}^{\fset/{\Tyv}}\bigl({\fset/{\Tyv}}(\Gamma,
  \argument) \times X_\vtype, \Sigmas_\ctype(X+Y)\bigr)$ whose component at $\Delta\in \fset/\Tyv$ sends $(h,e)\in \fset/\Tyv(\Gamma,\Delta)\times X_\vtype(\Delta)$ to $f_\Delta(x_1,\ldots,x_{\dom(\Gamma)},e)$, where $x_i\in X_{\Gamma(i)}(\Delta)$ is the image of $i\in \dom(\Gamma)$ under the map
\[ V_{\Gamma(i)}(\Gamma)\xto{V_{\Gamma(i)}h} V_{\Gamma(i)}(\Delta) \xto{\var_{\Gamma(i),\Delta}} X_{\Gamma(i)}(\Delta).\]
The clause thus specifies the desired labeled transition 
$\mathsf{lam}_\kappa\, x\colon \varphi.t \xto{e} t[e/x]$ in the model \eqref{eq:op-model-cbpv}.

\begin{remark}
 The need for substitution $\llangle-,-\rrangle$ as part of the behaviour functor \eqref{eq:cbpvbeh} comes (only) from the clause for $\lambda$-abstraction in the definition of $\rho^2$, where the substitution structure $f$ of the term $t$ is used to describe the behaviour of $\mathsf{lam}_\kappa\, x\colon \varphi.t$. This phenomenon occurs in the modelling of all languages with binders in higher-order abstract GSOS; see e.g.~\cite{gmstu23,gmtu24lics}. 
\end{remark}

\begin{remark} Under the prism of higher-order abstract GSOS, terms that do not $\beta$-reduce
are also assigned a dynamics via the component $\rho^{2}$. As before, these
assignments can be given in the form of labeled transitions. For example,
in a $\mathsf{prod}$ expression,
\[
  \mathsf{prod}_{\vtype}(v,-,-) \;\mapsto\; \{v\} \qquad
   \text{corresponds to} \qquad
  \frac{}{\mathsf{prod}(v) \xto{U} v}.
\]
In other words, a $\mathsf{prod}$ expression signals (via the distinguished
label $U$) that it is a producer with value $v$. This allows us to
implement rules such as
\[
  \frac{s \to \mathsf{prod}(v)}{\toexpr{s}{x}{t} \to
    \mathsf{app}((\mathsf{lam}\,x.t),v)}
  \quad \text{and} \quad
  \frac{}{\mathsf{force}(\mathsf{thunk}(t)) \to t}
\]
in the style of higher-order abstract GSOS as
\[
  \inference{s \xto{U} v}{\toexpr{s}{x}{t} \to
    \mathsf{app}((\mathsf{lam}\,x.t),v)} \quad \text{and} \quad
  \inference{s \xto{F} t}{\mathsf{force}(s) \to t}.
\]
\end{remark}

The operational model of $\rho$ is the coalgebra $
\gamma = \langle \gamma^{1},\gamma^{2} \rangle \c \Lambda
\to \llangle \Lambda,\Lambda\rrangle \times \Pow_\star B'(\Lambda,\Lambda) = B(\Lambda,\Lambda)$ 
given by \autoref{fig:gamma}. The following two propositions assert that the component $\gamma^{1}$ is
indeed the simultaneous substitution map $\mathrm{sub}^\flat$ and that
$\gamma^{2}$ accurately models $\beta$-reduction in \cbpv.

\begin{proposition}
  \label{prop:sub}
  For any well-typed term $\Gamma \vdash t \c \tau$, the following is true:
  \[
    \forall \vec{u} \in \prod_{i \in
      \dom(\Gamma)}\Lambda_{\Gamma(i)}(\Delta), \qquad \gamma^{1}(t)(\vec{u}) = t[\vec{u}].
  \]
\end{proposition}

\begin{proof}
We proceed by structural induction over $t$.
  We show the cases for variables and $\lambda$-abstraction, as the rest are
  similar. For a variable $\Gamma \vdash x \c \vtype$, we have that
  $\gamma^{1}(x)(\vec{u}) = \vec{u}_{x}$ by definition of $\gamma^{1}$,
  and $\vec{u}_{x} = x[\vec{u}]$ by definition of substitution. For a
  $\lambda$-abstraction $\mathsf{lam}\,x\c\vtype.M$, we have
  \[
    \begin{aligned}
      \gamma^{1}(\mathsf{lam}\,x\c\vtype.M)(\vec{u})
      & =
        \mathsf{lam}\,x\c\vtype.\gamma^{1}(M)(\vec{\up}^{\vtype}(\vec{u}),
        \oname{var}^{\vtype}(\oname{new})) \\
      &= \mathsf{lam}\,x\c\vtype.M[(\vec{\up}^{\vtype}(\vec{u}),
        \oname{var}^{\vtype}(\oname{new})] \\
      & = (\mathsf{lam}\,x\c\vtype.M)[\vec{u}]
    \end{aligned}
  \]
  where the first equality is by the definition of $\gamma^{1}$, the second by
  the inductive hypothesis and the third by the definition of simultaneous
  substitution on $\lambda$-abstractions.
\end{proof}

\begin{proposition}
  \label{prop:reduction}
  For any computation $\Gamma \vdash^{\one} t \c \ctype,\,\,t \to t' \iff t' \in
  \gamma^{2}(t)$.
\end{proposition}
\begin{proof}
We proceed by structural induction over $t$.
  We show the case for application, i.e.
  $t = \mathsf{app}_{\vtype,\ctype}(s,v)$ with $\Gamma \vdash^{\one}
  \mathsf{app}_{\vtype,\ctype}(s,v) \c \ctype$. The other cases are handled similarly.
  \begin{enumerate}
  \item $\mathsf{app}_{\vtype,\ctype}(s,v) \to t' \implies
    t' \in \gamma^{2}(\mathsf{app}_{\vtype,\ctype}(s,v))$. We
    identify two situations: first, $s \to s'$, thus
    $t' = \mathsf{app}_{\vtype,\ctype}(s',v)$ and second, $s = \mathsf{lam}\,x\c \vtype.M$,
    thus $t' = M[v/x]$. For the former, by the inductive hypothesis we have that
    $s' \in \gamma^{2}(s)$ and by definition of $\gamma^{2}$,
    $\mathsf{app}_{\vtype,\ctype}(s',v) \in
    \gamma^{2}(\mathsf{app}_{\vtype,\ctype}(s,v))$. For the
    latter, by the definition of $\gamma^{2}$,
    $\gamma^{1}(M)(x_{1},\dots,x_{\dom(\Gamma)},v) \in
    \gamma^{2}(\mathsf{app}_{\vtype,\ctype}(\lambda\,x\c \vtype.M, v))$. By
    \Cref{prop:sub}, $\gamma^{1}(M)(x_{1},\dots,x_{\dom(\Gamma)},v) =
    M[x_{1}/x_{1},\dots,x_{\dom(\Gamma)}/x_{\dom(\Gamma)},v/x]$, which is $M[v/x]$.
  \item $\gamma^{2}_{\ctype,\Gamma}(\mathsf{app}_{\vtype,\ctype}(s,v)) \ni t' \in
    \Lambda_{\ctype}(\Gamma) \implies \mathsf{app}_{\vtype,\ctype}(t,s) \to t'$.
    By the definition of $\gamma^{2}$ and the definition
    of $\rho^{2}$, we distinguish two cases: either there exists $s' \in
    \Lambda_{\arty{\vtype}{\ctype}}(\Gamma)$, with $s' \in \gamma^{2}(s)$ and thus
    $t' = \mathsf{app}_{\vtype,\ctype}(s',v)$, or there exists $d \in
    \Lambda_{\ctype}^{\Lambda_{\vtype}}(\Gamma)$ and $s' = d(v)$. In the first
    case, by the induction hypothesis we have that $s \to s'$ and thus
    $\mathsf{app}_{\vtype,\ctype}(s,v) \to \mathsf{app}_{\vtype,\ctype}(s',v)$.
    In the second case, by the definition of $\rho^{2}$, we infer that $s$ is of
    the form $\mathsf{lam}\,x\c \vtype.M$ and furthermore $d(v) = M[v/x]$. It suffices to show that
    $\mathrm{app}_{\vtype,\ctype}(\mathsf{lam}\,x\c \vtype.M,v) \to M[v/x]$,
    which is true. \qedhere
  \end{enumerate}
\end{proof}

\subsection{Operational Methods for \cbpv}
Next, we derive from the general framework the appropriate notions of applicative similarity and step-indexed logical relation for \cbpv. Recall that a relation $R\monoto \Lambda\times \Lambda$ can be presented as a family of relations $R_\tau(\Gamma)\seq \Lambda_\tau(\Gamma)\times \Lambda_\tau(\Gamma)$ indexed by $\tau\in \Ty$ and $\Gamma\in \fset/\Tyv$. Recall also from \autoref{sec:ho-gsos} the notion of congruence for functor algebras, and that congruences are not necessarily equivalence relations. Instantiating this to the initial $\Sigma$-algebra $\Lambda$ of \cbpv-terms, we obtain: 

\begin{definition}[Congruence for \cbpv]
  A relation $R \monoto \Lambda \times \Lambda$ is a \emph{congruence} if it is respected by all the operations of
  \cbpv:
  \begin{enumerate}
  \item $R_{\utype}(\Gamma)(\star,\star)$.
  \item For all variables $\Gamma \vdash^{\zero} x \c \vtype$, $R_{\vtype}(\Gamma)(x,x)$.
  \item For all values $\Gamma \vdash^{\zero} v_{1},v_{2} \c \vtype_{1}$,
    $R_{\vtype_{1}}(\Gamma)(v_{1},v_{2}) \implies R_{\vtype_{1} \boxplus
      \vtype_{2}}(\Gamma)(\linl_{\vtype_{1},\vtype_{2}}(v_{1}),
    \linl_{\vtype_{1},\vtype_{2}}(v_{2}))$.
  \item For all computations $\Gamma \vdash^{\one} t_{1},t_{2} \c \ctype$,
    $R_{\ctype}(\Gamma)(t_{1},t_{2}) \implies
    R_{U\ctype}(\Gamma)(\mathsf{thunk}_{\ctype}(t_{1}),
    \mathsf{thunk}_{\ctype}(t_{2}))$.
  \item For all values $\Gamma \vdash^{\zero} v_{1},v_{2} \c \vtype$,
    $R_{\vtype}(\Gamma)(v_{1},v_{2}) \implies
    R_{F\vtype}(\Gamma)(\mathsf{prod}_{\vtype}(v_{1}),
    \mathsf{prod}_{\vtype}(v_{2}))$.
  \item For all values $\Gamma \vdash^{\zero} v_{1},v_{2} \c \vtype$ and
    computations $\Gamma \vdash^{\one} t_{1},t_{2} \c \arty{\vtype}{\ctype}$,
    \[
      R_{\vtype}(\Gamma)(v_{1},v_{2}) \land R_{\arty{\vtype}{\ctype}}(\Gamma)(t_{1},t_{2}) \implies
    R_{{\ctype}}(\Gamma)(\mathsf{app}_{\vtype,\ctype}(t_{1},v_{1}),
    \mathsf{app}_{\vtype,\ctype}(t_{2},v_{2})).
  \]
  \item For all computations $\Gamma,x\c\vtype \vdash^{\one} M_{1},M_{2} \c
    \ctype$,
    \[
      R_{\ctype}(\Gamma + \check\vtype)(M_{1},M_{2}) \implies
      R_{\arty{\vtype}{\ctype}}(\Gamma)(\mathsf{lam}_{\ctype}\,x\c\vtype.M_{1},
    \mathsf{lam}_{\ctype}\,x\c\vtype.M_{2}).
    \]
  \end{enumerate}
Analogously for the remaining operations.
\end{definition}

\begin{definition}
 \emph{Contextual equivalence} is the greatest congruence $\contpre$ contained in $O\monoto \Lambda\times \Lambda$,
\[  O_\vtype(\Gamma) = \Lambda_\vtype(\Gamma)\times \Lambda_\vtype(\Gamma) \qqand O_\ctype(\Gamma) = \{ (t,s) \mid t{\Downarrow} \To s{\Downarrow} \},\]
where $t{\Downarrow}$ means that $t$ eventually $\beta$-reduces to a term that does not further reduce. (Since the language \cbpv is nondeterministic, this corresponds to \emph{may-termination}.)
\end{definition}
 As usual, $\contpre$ could alternatively defined in terms of contexts~\cite[Thm.~7.5.3]{10.5555/1076265}.

The weakening of the operational model
$\gamma = \langle \mathrm{sub}^\flat, \gamma^{2} \rangle \c \Lambda \to
\llangle \Lambda,\Lambda \rrangle \times \Pow_{\star}B'(\Lambda,\Lambda)$ is given by
\begin{equation}
  \label{eq:laxmodelcbpv}
  \widetilde{\gamma} = \langle \mathrm{sub}^\flat, \widetilde{\gamma}_2 \rangle \c
  \Lambda \to \llangle \Lambda,\Lambda \rrangle \times \Pow_{\star}B'(\Lambda,\Lambda),
\end{equation}
where $\widetilde{\gamma}_2 \c \Lambda \to \Pow_{\star}B'(\Lambda,\Lambda)$ is
defined as follows. On value types $\vtype$,
$(\widetilde{\gamma}^{2})_{\vtype} = ({\gamma}^{2})_{\vtype}$. On computation
types $\ctype$, we put $(\tilde{\gamma}_2)_{\kappa,\Gamma} = \{ t'\mid \Gamma \vdash t'\colon \kappa \text{ and } t\To t' \}$, where $t\To t'$ means that $t$ eventually $\beta$-reduces to~$t'$ (that is, $\To$ is the reflexive transitive closure of $\to$).

Finally, we pick a suitable relation lifting for the bifunctor $B(X,Y)= \llangle X,Y
\rrangle \times \Pow_{\star}B'(X,Y)$ \eqref{eq:cbpvbeh}:
\begin{equation}
  \label{eq:eglimilnercbpv}
  \barB(R,S) = \ol{\llangle} R,S \ol{\rrangle} \times \arP_\star\barB'(R,S)
\end{equation}
where $\ol{\llangle} \argument,\argument \ol{\rrangle}$ and $\barB'$
are the canonical liftings of the bifunctors $\llangle \argument,\argument
\rrangle$ and $\barB'$, and $\arP_\star$ is the pointwise left-to-right Egli-Milner
relation lifting.

We are now in a position to define applicative similarity for \cbpv. Spelling out \Cref{def:app-sim} for the weakening $\wt{\gamma}$
\eqref{eq:laxmodelcbpv} and the relation lifting $\barB$
\eqref{eq:eglimilnercbpv} yields the following notion:

\begin{definition}[Applicative similarity for \cbpv]\label{prop:weak-sim-cbpv}
A relation $R \monoto
\Lambda \times \Lambda$ is a \emph{weak simulation} if for all pairs of
terms $\Gamma \vdash t,s \c \tau$ such that $R_{\tau}(\Gamma)(t,s)$, the
following hold:
\begin{enumerate}
\item For all substitutions $\vec{u} \c \prod_{i \in
    \dom(\Gamma)}\Lambda_{\Gamma(i)}(\Delta)$, we have that $R_{\tau}(\Delta)(t[\vec{u}], s[\vec{u}])$.
\item If $\tau = \ctype$ and $t \to t'$, then $s \To s'$ and
  $R_{\ctype}(\Gamma)(t',s')$.
\item If $t = \star$, then $s = \star$.
\item If $t = \mathsf{thunk}_{\ctype}(t')$, then $s =
  \mathsf{thunk}_{\ctype}(s')$ and $R_{\tau}(\Gamma)(t',s')$.
\item If $t = \linl_{\vtype_{1},\vtype_{2}}(v)$, then $s =
  \linl_{\vtype_{1},\vtype_{2}}(w)$ and $R_{\vtype_{1}}(\Gamma)(v,w)$.
\item If $t = \linr_{\vtype_{1},\vtype_{2}}(v)$, then $s =
  \linr_{\vtype_{1},\vtype_{2}}(w)$ and $R_{\vtype_{2}}(\Gamma)(v,w)$.
\item If $t = \mathsf{pair}_{\vtype_{1},\vtype_{2}}(v_{1},v_{2})$, then $s =
  \mathsf{pair}_{\vtype_{1},\vtype_{2}}(w_{1},w_{2})$ with
  $R_{\vtype_{1}}(\Gamma)(v_{1},w_{1})$ and
  $R_{\vtype_{2}}(\Gamma)(v_{2},w_{2})$.
\item If $t = \mathsf{pair}_{\ctype_{1},\ctype_{2}}(t_{1},t_{2})$, then $s \To
  \mathsf{pair}_{\ctype_{1},\ctype_{2}}(s_{1},s_{2})$ with
  $R_{\ctype_{1}}(\Gamma)(t_{1},s_{1})$ and
  $R_{\ctype_{2}}(\Gamma)(t_{2},s_{2})$.
\item If $t = \mathsf{fold}_{\ctype}(t')$ then $s \To \mathsf{fold}_{\ctype}(s')$
  and $R_{\ctype[\mu \alpha.\ctype/\alpha]}(\Gamma)(t',s')$.
\item If $t = \lambda x \c \vtype.M$, then $s \To \lambda x \c \vtype.N$ and for
  all $\Gamma \vdash^{\zero} v \c \vtype$, $R_{\ctype}(\Gamma)(M[v/x],N[v/x])$.
\item If $t = \mathsf{prod}_{\vtype}(v)$ then $s \To \mathsf{prod}_{\vtype}(w)$
  and $R_{\vtype}(\Gamma)(v,w)$.
\end{enumerate}
\emph{Applicative
  similarity} $\appsim$ is the (pointwise) union of all applicative
simulations.
\end{definition}

Similarly, \autoref{def:abslogrel} yields the following notion of step-indexed logical relation:

\begin{definition}[Step-indexed logical relation for \cbpv]
  \label{eq:logrelcbpv}
The \emph{step-indexed logical
relation} $\logrel = \bigcap_\alpha \logrel^\alpha$ is given by the relations $\logrel^{\alpha} \monoto \Lambda\times \Lambda$ 
defined by transfinite induction as follows: put $\logrel^{0}= \Lambda\times \Lambda$ and $\logrel^\alpha = \bigcap_{\beta<\alpha} \logrel^\beta$ for limit ordinals $\alpha$. For successor ordinals, given terms $\Gamma \vdash t,s \c \tau$ be terms in \cbpv we put $\logrel^{\alpha+1}_{\type}(t,s)$ if and only if the
  following hold:
  \begin{enumerate}
  \item For all pairs of substitutions $\vec{v}, \vec{w} \c \prod_{i \in
      \dom(\Gamma)}\Lambda_{\Gamma(i)}(\Delta)$ whose subterms are
    pairwise related in $\logrel^{\alpha}$, we have that
    $\logrel^{\alpha}_{\tau}(\Delta)(t[\vec{u}], s[\vec{w}])$.
  \item If $\tau = \ctype$ and $t \to t'$, then $s \To s'$ and
    $\logrel^{\alpha}_{\ctype}(\Gamma)(t',s')$.
  \item If $t = \star$, then $s = \star$.
  \item If $t = \mathsf{thunk}_{\ctype}(t')$, then $s =
    \mathsf{thunk}_{\ctype}(s')$ and $\logrel^{\alpha}_{\tau}(\Gamma)(t',s')$.
  \item If $t = \linl_{\vtype_{1},\vtype_{2}}(v)$, then $s =
    \linl_{\vtype_{1},\vtype_{2}}(w)$ and $\logrel^{\alpha}_{\vtype_{1}}(\Gamma)(v,w)$.
  \item If $t = \linr_{\vtype_{1},\vtype_{2}}(v)$, then $s =
    \linr_{\vtype_{1},\vtype_{2}}(w)$ and $\logrel^{\alpha}_{\vtype_{2}}(\Gamma)(v,w)$.
  \item If $t = \mathsf{pair}_{\vtype_{1},\vtype_{2}}(v_{1},v_{2})$, then $s =
    \mathsf{pair}_{\vtype_{1},\vtype_{2}}(w_{1},w_{2})$ with
    $\logrel^{\alpha}_{\vtype_{1}}(\Gamma)(v_{1},w_{1})$ and
    $\logrel^{\alpha}_{\vtype_{2}}(\Gamma)(v_{2},w_{2})$.
  \item If $t = \mathsf{pair}_{\ctype_{1},\ctype_{2}}(t_{1},t_{2})$, then $s \To
    \mathsf{pair}_{\ctype_{1},\ctype_{2}}(s_{1},s_{2})$ with
    $\logrel^{\alpha}_{\ctype_{1}}(\Gamma)(t_{1},s_{1})$ and
    $\logrel^{\alpha}_{\ctype_{2}}(\Gamma)(t_{2},s_{2})$.
  \item If $t = \mathsf{fold}_{\ctype}(t')$ then $s \To \mathsf{fold}_{\ctype}(s')$
    and $\logrel^{\alpha}_{\ctype[\mu \alpha.\ctype/\alpha]}(t',s')$.
  \item If $t = \lambda x \c \vtype.M$, then $s \To \lambda x \c \vtype.N$ and for
    all $\Gamma \vdash^{\zero} v,w \c \vtype$,
    \[
      \logrel^{\alpha}_{\vtype}(\Gamma)(v,w) \implies
      \logrel^{\alpha}_{\ctype}(\Gamma)(M[v/x],N[w/x]).
    \]
  \item If $t = \mathsf{prod}_{\vtype}(v)$ then $s \To \mathsf{prod}_{\vtype}(w)$
    and $\logrel^{\alpha}_{\vtype}(v,w)$.
  \end{enumerate}
\end{definition}
The Abstract Soundness Theorem \ref{thm:congruence-ho-abstract-gsos} now yields as a special case:
\begin{theorem}[Soundness Theorem for \cbpv] Both applicative similarity $\appsim$ and the step-indexed logical relation $\logrel$ are sound for the contextual preorder: for all terms $p,q\in \Lambda_\tau(\Lambda)$,
\[
p \appsim q \; \implies\; p\contpre q\qquad\text{and}\qquad
\logrel(p,q) \; \implies\; p\contpre q.\]
\end{theorem}

\begin{proof}
Clearly both $\appsim$ and $\logrel$ are contained in $O$. Therefore, we only need to check the conditions \ref{asm:1}--\ref{asm:6} of \autoref{thm:congruence-ho-abstract-gsos}. Condition \ref{asm:1} holds because $+$, $\times$ and $\delta$ (being a left adjoint) preserve directed colimits, strong epis and monos. The order on $(\Set^{\fset/{\Tyv}})^{\Tyl}(Z,\llangle X,Y \rrangle \times \Pow_{\star}B'(X,Y))$ demanded by \ref{asm:2} is given be pointwise equality in the first component and pointwise inclusion in the second component. Due to the use of the left-to-right Egli-Milner lifting, \ref{asm:4} holds. \ref{asm:3} and \ref{asm:5} follow like for $\xCL$, since our lifting $\barB$ can be expressed as a canonical lifting of a deterministic behaviour bifunctor composed with the left-to-right lifting of the power set functor. In particular \ref{asm:5} holds because the rules have no negative premises. The lax bialgebra condition \ref{asm:6} again means that the operational rules remain sound when strong transitions are replaced by weak ones. This is easily verified by inspecting the rules; for instance, the weak versions of the rules for
  $\lfst$ are
  \[
    \inference{t \To t'}{\lfst(t) \To \lfst(t')}
    \qquad
    \inference{t \To \mathsf{pair}(t',s)}{\lfst(t) \To t'}
  \]
and they are clearly sound because they follow by repeated application of the strong versions.
\end{proof}

\begin{example}
  \label{ex:cbpvbetalaw}
  One of the $\beta$-laws for \cbpv states that
\[   t \contpre \mathsf{force}(\mathsf{thunk}(t)) \contpre t\qquad \text{for all computations $t\colon \ctype$}.\]
To prove this, we show (i) $\logrel^\alpha_{\ctype}(t,\mathsf{force}(\mathsf{thunk}(t))$ and (ii) $\logrel^\alpha_\ctype(\mathsf{force}(\mathsf{thunk}(t),t)$ by transfinite induction. The only non-trivial case is the successor step. For (i), suppose that $t\to t'$. Then $\mathsf{force}(\mathsf{thunk}(t))\To t'$ and $\logrel^\alpha(t',t')$ because $\logrel^\alpha$ is reflexive (being a congruence on an initial algebra). Similarly, for (ii), we have  $\mathsf{force}(\mathsf{thunk}(t))\to t$ and $t\To t$ and $\logrel^\alpha(t,t)$.
\end{example}

\begin{example}
  In this example we investigate one half of the ``thunk'' $\eta$-law for \cbpv:
  \[
    v \contpre \mathsf{thunk}(\mathsf{force}(v))\qquad \text{for all closed values $v \c U\ctype$.}
  \]
  To prove this, we show 
  $\logrel^\alpha_{U\ctype}(v,\mathsf{thunk}(\mathsf{force}(v))$ by transfinite induction. The only non-trivial case is the successor step $\alpha\to \alpha+1$. We show that
  $\logrel_{U\ctype}^{\alpha+1}(v,\mathsf{thunk}(\mathsf{force}(v))$. Since $v$
  is closed, it is not a variable. Thus $v = \mathsf{thunk}(t)$, meaning that $\gamma^{2}(v) = \{t\}$. As
  $\wt{\gamma}_2(\mathsf{thunk}(\mathsf{force}(\mathsf{thunk}(t))=
  \{\mathsf{force}(\mathsf{thunk}(t)\})$, it suffices to show
  $\logrel_{\ctype}^{\alpha}(t,\mathsf{force}(\mathsf{thunk}(t))$, which is
  true by \Cref{ex:cbpvbetalaw}.
\end{example}

\begin{remark}
  To capture the reverse direction
  $\mathsf{thunk}(\mathsf{force}(v)) \contpre v$ via the logical relation
  $\logrel$, one has to define the more liberal \emph{testing} weakening as in \cite[Def. 3.10, Ex. 4.34]{gmtu24lics}. More precisely,
  for each value $\Gamma \vdash^{\zero} v \c U\ctype$, the ``testing''
  weakening $\wt{\gamma}_2(v)$ additionally includes the transition
  $\mathsf{force}(v)$, which then allows the logical relation to identify
  $\mathsf{thunk}(\mathsf{force}(v))$ with $v$.
\end{remark}

\section{Conclusions and Future Work}

We developed a novel theory of compositional program equivalences for fine-grain call-by-value (FGCBV) and call-by-push-value (CPBV) languages. In doing so, we demonstrated that the categorical
framework of higher-order abstract GSOS is capable of handling call-by-value semantics, which had been an open question up to this point, and derived congruence properties of FGCPV and CBPV from general congruence results available in the abstract framework. The insight that higher-order abstract GSOS applies smoothly to such complex languages, without any need for adapting or extending the existing abstract theory of congruence, is somewhat remarkable and further highlights the scope and expressive power of the framework.

Effects play a prominent role in call-by-push-value, and may go beyond
the nondeterministic example of \cbpv. Of particular interest would be to model
effect-parameterized variants of CBPV languages to capture, e.g., concurrent or stateful computations. The higher-order abstract GSOS framework is particularly
useful for developing unifying approaches to program reasoning, although the
existing theory is not yet compatible with computational effects at their full
generality.

There are also other interesting aspects of CBPV, including its use as
an intermediate language for e.g. call-by-name and call-by-value languages. We
plan to further study the properties of FGCBV and CBPV languages, with a focus
on their use as compilation targets, and understand
them at a high level of abstraction similarly to the generic normalization
theorem of \cite{gmstu24}. For instance, we are interested
in generalizing the embeddings of CBV and CBN into CBPV to a broader
scope, and to formally capture the \emph{families} of languages that embed into CBPV.

The current work exemplifies the effectiveness of the higher-order
abstract GSOS framework towards modelling the operational semantics of
higher-order languages with binding and substitutions and reasoning about program
equivalences. The unifying pattern emerging through all of the examples is that the
mathematical structures of such languages live in presheaf toposes. We wish to
develop \emph{rule formats} for the operational
semantics of languages with binding and substitution, which would
leverage the well-behavedness of such categories, and provide the means to
develop semantics without the extra overhead of working directly with the
categorical structures.

\begin{acks}
Sergey Goncharov
acknowledges funding by the Deutsche Forschungsgemeinschaft (DFG, German
  Research Foundation) -- project numbers 527481841 and 501369690. Stelios Tsampas acknowledges funding by the Deutsche Forschungsgemeinschaft (DFG, German
  Research Foundation) -- project number 527481841. Henning Urbat
acknowledges funding by by the Deutsche Forschungsgemeinschaft (DFG, German
  Research Foundation) -- project number 470467389.
\end{acks}
\bibliography{mainBiblio}

\ifarx{
\clearpage
\appendix

\section{Higher-Order GSOS Law for CBPV}\label{app:rho-cbpv}
The operational semantics of the language \cbpv is modeled by a $V$-pointed
higher-order GSOS law $\rho$ of $\Sigma$ \eqref{eq:cbpvsyn} over $B$
\eqref{eq:cbpvbeh} as follows. The component
\[\rho_{(X,\var),Y}\colon \Sigma(X \times \llangle X,Y \rrangle
    \times \Pow_{\star}B'(X,Y)) \to
    \llangle X,\Sigmas(X+Y)\rrangle \times \Pow_{\star}B'(X,\Sigma^{\star}(X + Y)) \] 
is defined as a pairing $\langle\rho^{1}_{(X,\var),Y} \comp
\Sigma p,\,\rho^{2}_{(X,\var),Y}\rangle$, where $p$ is the middle product projection and
\begin{equation}
  \label{eq:rhorho}
    \begin{aligned}
  & \rho^{1}_{(X,\var),Y} \c \Sigma\llangle X,Y \rrangle \to
    \llangle X, \Sigma^{\star}(X + Y) \rrangle, \\
  & \rho^{2}_{(X,\var),Y} \c \Sigma(X \times \llangle X,Y \rrangle
    \times \Pow_{\star}B'(X,Y)) \to
    \Pow_{\star}B'(X,\Sigma^{\star}(X + Y)).
  \end{aligned}
\end{equation}
Here $\rho^{1}$ models the simultaneous substitution structure in the sense that it induces the map $\mathrm{sub}^{\flat} \c
\Lambda \to \llangle \Lambda,\Lambda\rrangle$, and $\rho^2$ models the dynamics of \cbpv. Formally, the two components are defined as follows.

\medskip\noindent {\textbf{Definition of $\rho^{1}_{(X,\var),Y}$:}}\\
The component of the natural transformation $\rho^{1}_{(X,\var),Y}$ at $\tau\in\Ty$ and $\Gamma\in \fset/\Tyv$, i.e.
\[ \rho^{1}_{(X,\var),Y,\tau,\Gamma} \c \Sigma_\tau\llangle X,Y \rrangle(\Gamma) \to
    \Set^{\fset/{\Tyv}}\bigl(\prod_{i \in
    \dom(\Gamma)}X_{\Gamma(i)}, \Sigmas_{\tau}(X+Y)\bigr), \]
is defined below. Here, we fix a context $\Delta\in \fset/\Phi$ and use ``curried'' notation; for
instance, the assignment $\linl_{\vtype_1,\vtype_2}(f),\vec{u} \mapsto \linl_{\vtype_1,\vtype_2}(f_\Delta(\vec{u}))$ below means that
$\rho^{1}_{(X,\var),Y,\vtype_{1} \boxplus \vtype_{2},\Gamma}(\linl_{\vtype_1,\vtype_2}(f))$ is the natural transformation $\prod_{i \in
    \dom(\Gamma)}X_{\Gamma(i)} \to \Sigmas_{\tau}(X+Y)$ whose component at $\Delta$ is given by
 $\lambda \vec{u}.\linl(f_\Delta(\vec{u}))$). We let $\vec{u}_i\in X_{\Gamma(i)}(\Delta)$ denote the $i$-th entry of $\vec{u}\in \prod_{i \in
    \dom(\Gamma)}X_{\Gamma(i)}(\Delta)$.

 Then for the non-binding operators of the signature, the map $\rho^{1}_{(X,\var),Y,\tau,\Gamma}$ is given by:
\begin{alignat*}{4}
& x,\vec{u}
    && \quad \mapsto \quad && \vec{u}_{x} &&\qquad\qquad (x\in V_\tau(\Gamma), \text{ i.e. } \Gamma(x)=\tau \})\\
&\star,\vec{u}
   && \quad \mapsto \quad && \star && \\
& \mathsf{app}_{\vtype,\ctype}(f,g),\vec{u}
    &&\quad\mapsto\quad && \mathsf{app}_{\vtype,\ctype}(f_\Delta(\vec{u}),g_\Delta(\vec{u}))  &&\qquad\qquad (f\in \llangle X,Y\rrangle_{\arty{\vtype}{\ctype}}(\Gamma),\, g\in \llangle X,Y\rrangle_{\vtype}(\Gamma))\\
&f\oplus_\ctype g,\vec{u}
    && \quad\mapsto\quad && f_\Delta(\vec{u}) \oplus_\ctype g_\Delta(\vec{u}) &&\qquad\qquad (f,g\in \llangle X,Y\rrangle_\ctype(\Gamma))\\
& \mathsf{force}_\ctype(f),\vec{u}
    &&\quad\mapsto\quad && \mathsf{force}_\ctype(f_\Delta(\vec{u})) &&\qquad\qquad (f\in \llangle X,Y\rrangle_{U\ctype}(\Gamma))\\
&\mathsf{thunk}_\ctype(f),\vec{u}
    && \quad \mapsto \quad && \mathsf{thunk}_\ctype(f_\Delta(\vec{u})) &&\qquad\qquad (f\in \llangle X,Y\rrangle_\ctype(\Gamma))\\
& \mathsf{prod}_\vtype(f),\vec{u}
    &&\quad\mapsto\quad&& \mathsf{prod}_\vtype(f_\Delta(\vec{u})) &&\qquad\qquad (f\in \llangle X,Y\rrangle_\vtype(\Gamma))\\
& \linl_{\vtype_1,\vtype_2}(f),\vec{u}
    &&\quad\mapsto\quad && \linl_{\vtype_1,\vtype_2}(f_\Delta(\vec{u})) &&\qquad\qquad (f\in \llangle X,Y\rrangle_{\vtype_1}(\Gamma)) \\
 & \linr_{\vtype_1,\vtype_2}(f),\vec{u}
    && \quad\mapsto\quad && \linr_{\vtype_1,\vtype_2}(f_\Delta(\vec{u})) &&\qquad\qquad (f\in \llangle X,Y\rrangle_{\vtype_2}(\Gamma))\\
& \mathsf{pair}_{\vtype_{1},\vtype_{2}}(f,g),\vec{u}
    &&\quad\mapsto\quad && \mathsf{pair}_{\vtype_{1},\vtype_{2}}(f_\Delta(\vec{u}),g_\Delta(\vec{u})) &&\qquad\qquad (f\in \llangle X,Y\rrangle_{\vtype_1}(\Gamma),\, g\in \llangle X,Y\rrangle_{\vtype_2}(\Gamma))\\
& \mathsf{pair}_{\ctype_{1},\ctype_{2}}(f,g),\vec{u}
    &&\quad\mapsto\quad && \mathsf{pair}_{\ctype_{1},\ctype_{2}}(f_\Delta(\vec{u}),g_\Delta(\vec{u})) && \qquad\qquad (f\in \llangle X,Y\rrangle_{\ctype_1}(\Gamma),\, g\in \llangle X,Y\rrangle_{\ctype_2}(\Gamma))\\
& {\mathsf{fst}}_{\ctype_1,\ctype_2}(f),\vec{u}
    &&\quad\mapsto\quad && {\mathsf{fst}}_{\ctype_1,\ctype_2}(f_\Delta(\vec{u})) && \qquad\qquad (f\in \llangle X,Y\rrangle_{\ctype_1\otimes \ctype_2}(\Gamma))\\
    &{\mathsf{snd}}_{\ctype_1,\ctype_2}(f),\vec{u}
    && \quad\mapsto\quad&& {\mathsf{snd}}_{\ctype_1,\ctype_2}(f_\Delta(\vec{u})) &&\qquad\qquad (f\in \llangle X,Y\rrangle_{\ctype_1\otimes \ctype_2}(\Gamma))\\
& \mathsf{fold}_\ctype(f),\vec{u}
    &&\quad\mapsto\quad && \mathsf{fold}_\ctype(f_\Delta(\vec{u})) &&\qquad\qquad (f\in \llangle X,Y\rrangle_{\ctype[\mu\alpha.\ctype/\alpha]}(\Gamma))\\
&\mathsf{unfold}_\ctype(f),\vec{u}
    && \quad\mapsto\quad && \mathsf{unfold}_\ctype(f_\Delta(\vec{u})) &&\qquad\qquad (f\in \llangle X,Y\rrangle_{\mu\alpha.\ctype}(\Gamma))
\end{alignat*}
For the binding operators, the map $\rho^{1}_{(X,\var),Y,\tau,\Gamma}$ is given as follows. For $\lambda$-abstractions,
\[  \mathsf{lam}_\ctype\,x\c\vtype. f,\vec{u}
    \quad \mapsto \quad \mathsf{lam}_\ctype\,x\c\vtype.f_{\Delta+\check\vtype}(\vec{\up}_{X,\Delta}^\vtype(\vec{u}), \var_{\Delta+\check\vtype,\vtype}(\oname{new}_{\Delta}^{\vtype})) \]
where $f\in \delta^\vtype\llangle X,Y\rrangle_\ctype(\Gamma) = \llangle X,Y\rrangle_\kappa(\Gamma+\check\vtype)$. In more detail, we are given a natural transformation
\[ f\colon (\prod_{i\in \dom(\Gamma)} X_{\Gamma(i)}) \times X_{\vtype} = \prod_{i\in \dom(\Gamma+\check\vtype)} X_{\Gamma(i)} \to Y_\kappa.  \]
We turn $\vec{u}\in \prod_{i\in \dom(\Gamma)} X_{\Gamma(i)}(\Delta)$ into 
\[ \vec{\up}_{X,\Delta}^\vtype(\vec{u})\in \prod_{i\in \dom(\Gamma)} X_{\Gamma(i)}(\Delta+\check\vtype) \]
by applying the map $(\up_{X,\Delta}^\vtype)_{\Gamma(i)} \colon X_{\Gamma(i)}(\Delta)\to X_{\Gamma(i)}(\Delta+\check\vtype)$ of \eqref{eq:up} in the $i$th component. Moreover, 
\[ \var_{\Delta+\check\vtype,\vtype}(\oname{new}_{\Delta}^{\vtype}) \in X_{\vtype}(\Delta+\check\vtype)   \]
where $\var_{\Delta+\check\vtype,\vtype}\colon V_{\vtype}(\Delta+\check\vtype) \to X_{\vtype}(\Delta+\check\vtype)$ is the component of the point $\var\colon V\to X$ of $X$, and we regard $\oname{new}_{\Delta}^\vtype\colon \check 1\cong \vtype\to \Delta+\check\vtype$ as an element of $V(\Delta+\check\vtype)$. By applying the $(\Delta+´\check\vtype)$-component of $f$ to $(\vec{\up}_{X,\Delta}^\vtype(\vec{u}), \var_{\Delta+\check\vtype,\vtype}(\oname{new}_{\Delta}^{\vtype}))$ we thus obtain
\[ f_{\Delta+\check\vtype}(\vec{\up}_{X,\Delta}^\vtype(\vec{u}), \var_{\Delta+\check\vtype,\vtype}(\oname{new}_{\Delta}^{\vtype})) \in Y_\ctype(\Delta+\check\vtype), \]
and so 
\[ \mathsf{lam}_\ctype\,x\c\vtype.f_{\Delta+\check\vtype}(\vec{\up}_{X,\Delta}^\vtype(\vec{u}), \var_{\Delta+\check\vtype,\vtype}(\oname{new}_{\Delta}^{\vtype})) \in (\Sigmas_{\arty{\vtype}{\ctype}}(X+Y))(\Delta). \]
For the remaining binding operators, we have (omitting subscripts for better readability)
\[  \toexprt{g}{x}{f}{\vtype}{\ctype},\vec{u}
    \quad\mapsto\quad \toexprt{g(\vec{u})}{x}{f(\vec{\up}^{\vtype}(\vec{u}),
      \oname{var}_{\vtype}(\oname{new}))}{\vtype}{\ctype},  \]
and 
\[
 \mathsf{case}_{\vtype_1,\vtype_1,\ctype}(f,g,h),\vec{u}
    \quad \mapsto \quad \mathsf{case}_{\vtype_1,\vtype_1,\ctype}(f(\vec{u}),g(\vec{\up}^{\vtype_{1}}(\vec{u}),
      \oname{var}_{\vtype_{1}}(\oname{new})),
      h(\vec{\up}^{\vtype_{2}}(\vec{u}),\oname{var}_{\vtype_{2}}(\oname{new})))\\
\]
and 
\[
 \mathsf{pm}_{\vtype_1,\vtype_1,\ctype}(f,g),\vec{u}
    \mapsto \mathsf{pm}_{\vtype_1,\vtype_1,\ctype}(f(\vec{u}),
    g(\delta^{\vtype_{1}}(\vec{\up}^{\vtype_{2}})
      \comp \vec{\up}^{\vtype_{1}}(\vec{u}),
      \oname{var}_{\vtype_{1}}(\oname{old} \comp \oname{new}),
      \oname{var}_{\vtype_{2}}(\oname{new} \comp \oname{new}))).
\]
This concludes the definition of $\rho^{}_{(X,\var),Y}$.

\medskip\noindent\textbf{Definition of $\rho^{2}_{(X,\var),Y}$:}\\
The component of the natural transformation $\rho^{2}_{(X,\var),Y}$ at $\tau\in \Ty$ and $\Gamma\in\fset/\Tyv$ is the map
\[    \rho^{2}_{(X,\var),Y,\tau,\Gamma} \c
     \Sigma_\tau(X \times \llangle X,Y \rrangle
      \times \Pow_{\star}B'(X,Y))(\Gamma)
     \to \Pow\cdot B'_\tau(X,\Sigma^{\star}(X + Y))(\Gamma)
\]
given by
{\allowdisplaybreaks
  \begin{alignat*}{4}
    &{x}
    && \quad\mapsto\quad && \varnothing \\
    & (x\in V_{\vtype}(\Gamma)) &&&& \\[1ex]
    &{\star}
    && \quad\mapsto\quad &&1\\[1ex]
    &{\linl_{\vtype_1,\vtype_2}(v,-,-)}
    && \quad\mapsto\quad && \{\inl(v)\} \\
    & (v\in X_{\vtype_1}(\Gamma))\\[1ex]
    &{\linr_{\vtype_1,\vtype_2}(v,-,-)}
    && \quad\mapsto\quad && \{\inr(v)\} \\
    & (v\in X_{\vtype_2}(\Gamma))\\[1ex]
    &{\mathsf{thunk}_{\ctype}(t,-,-)}
    && \quad\mapsto\quad && \{t\} \\
    & (t\in X_\ctype(\Gamma))\\[1ex]
    &{\mathsf{pair}_{\vtype_{1},\vtype_{2}}((v,-,-),(w,-,-))}
    && \quad\mapsto\quad && \{(v,w)\} \\
    & (v\in X_{\vtype_1}(\Gamma),\, w\in X_{\vtype_2}(\Gamma))\\[1ex]
    &{\mathsf{fold}_{\ctype}(t,-,-)}
    && \quad\mapsto\quad && \{t\} \\
    & (t\in X_{\ctype}(\Gamma))\\[1ex]
    &{\mathsf{unfold}_{\ctype}(-,-,L)}
    && \quad\mapsto\quad && \{t' \mid t' \in L \cap Y_{\ctype[\mu \alpha.\ctype/\alpha]}\} \\
    & (L\in \Pow\cdot B'_{\ctype[\mu \alpha.\ctype/\alpha]}(X,Y)(\Gamma))\\[1ex]
    &{\mathsf{lam}_{\ctype}\,x\c
      \vtype.\,(-,f,-)}
    && \quad\mapsto\quad && \{\lambda e. f(x_{1},\dots,x_{\dom(\Gamma)},e)\} \text{ (see below)} \\
    & (f\in \delta^{\vtype}\llangle X,Y\rrangle_\ctype(\Gamma)=\llangle X,Y\rrangle_\ctype(\Gamma+\check\vtype))\\[1ex]
    &{(t,-,-) \oplus_\ctype (s,-,-)}
    && \quad\mapsto\quad && \{t,s\} \\
    & (t,s\in X_{\ctype}(\Gamma))\\[1ex]
    &\mathsf{case}_{\vtype_{1},\vtype_{2},\ctype}((-,-,L),(-,-,-),(-,-,-))
    && \quad \mapsto\quad && \{ \mathsf{app}_{\vtype_1,\ctype}(\mathsf{lam}\,x\c\vtype_{1}.s,t) \mid t \in L
       \cap Y_{\vtype_1}(\Gamma)\}\,{\cup} \\
    & (L\in \Pow\cdot B'_{\vtype_1\boxplus \vtype_2}(X,Y)(\Gamma))   &&&& \{\mathsf{app}_{\vtype_2,\ctype}(\mathsf{lam}\,x\c\vtype_{2}.r,t) \mid t \in L \cap
        Y_{\vtype_2}(\Gamma)\}\\[1ex]
    &{\lfst_{\ctype_1,\ctype_2}(-,-,L)}
    && \quad\mapsto\quad && \{\lfst_{\ctype_1,\ctype_2}(t') \mid t' \in L \cap
       Y_{\ctype_1\boxtimes\ctype_2}(\Gamma)\}\,{\cup}\\
    & (L\in \Pow\cdot B'_{\ctype_1\otimes \ctype_2}(X,Y)(\Gamma)) &&&& \{\fst(t') \mid t' \in L \cap
        Y_{\ctype_1}\times Y_{\ctype_2}(\Gamma)\} \\[1ex]
    &{\lsnd_{\ctype_1,\ctype_2}(-,-,L)}
    && \quad\mapsto\quad && \{\lsnd_{\ctype_1,\ctype_2}(t') \mid t' \in L \cap
       Y_{\ctype_1\boxtimes\ctype_2}(\Gamma)\}\,{\cup}\\
    & (L\in \Pow\cdot B'_{\ctype_1\otimes \ctype_2}(X,Y)(\Gamma))  &&&& \{\snd(t') \mid t' \in L \cap Y_{\ctype_1}\times Y_{\ctype_2}(\Gamma)\} \\[1ex]
    &{\mathsf{pair}_{\ctype_{1},\ctype_{2}}((t,-,-),(s,-,-))}
    && \quad\mapsto\quad && \{(t,s)\}\\
 & (t\in X_{\ctype_1}(\Gamma),\, s\in X_{\ctype_2}(\Gamma))\\[1ex]
    &{\mathsf{app}_{\vtype,\ctype}((-,-,L),(v,-,-))}
    && \quad\mapsto\quad && \{ f(v) \mid f \in L \cap
       Y_{\ctype}^{X_{\vtype}}(\Gamma)\}\,{\cup}\\
 & (L\in \Pow\cdot B_{\arty{\vtype}{\ctype}}'(X,Y)(\Gamma),\, v\in X_\vtype(\Gamma))
    &&&& \{\mathsf{app}_{\vtype,\ctype}(t', v) \mid t' \in L \cap
        Y_{\arty{\vtype}{\ctype}}(\Gamma)\} \\[1ex]
    &{\toexprt{(-,-,L)}{x}{(t,-,-)}{\vtype}{\ctype}}
    && \quad\mapsto\quad &&
       \{\mathsf{app}_{\vtype,\ctype}(\mathsf{lam}\,x\c\vtype.t, t')
       \mid t' \in L \cap Y_{\vtype}(\Gamma)\}\,\cup \\
    & (L\in \Pow\cdot B'_{F\vtype}(X,Y)(\Gamma),\, t\in \delta^{\vtype}X(\Gamma))  &&&& \{\toexprt{t'}{x}{t}{\vtype}{\ctype} \mid t' \in L \cap Y_{F\vtype}(\Gamma)\}\\[1ex]
    &{\mathsf{force}_{\ctype}(-,-,L)}
    && \quad\mapsto\quad && L\cap Y_{\ctype} \\
    & (L\in \Pow\cdot B'_{U\ctype}(X,Y)(\Gamma))  \\[1ex]
    &{\mathsf{prod}_{\vtype}(v,-,-)}
    && \quad\mapsto\quad && \{v\}\\
    & (v\in X_{\vtype}) \\[1ex]
    &{\mathsf{pm}_{\vtype_{1},\vtype_{2},\ctype}((-,-,L),(s,-,-))}
    && \quad\mapsto\quad && \{\mathsf{app}((\mathsf{lam}\,y.\mathsf{app}((\mathsf{lam}\,x.s),\fst(t))),\snd(t)) \mid \\
    & (L\in \Pow\cdot B_{\vtype_2\otimes \vtype_2}(X,Y)(\Gamma),\, s\in \delta^{\vtype_1}\delta^{\vtype_2}X(\Gamma)) &&&& \;\; t \in L \cap Y_{\vtype_1}(\Gamma) \times Y_{\vtype_{2}}(\Gamma)\}
\end{alignat*}
\begin{remark}\label{rem:rho-lambda-def}
In the clause for $\mathsf{lam}_{\ctype}\,x\c
      \vtype.\,(t,f,L)$, the expression $\lambda e. f(x_{1},\dots,x_{\dom(\Gamma)},e)$ on the right-hand side denotes the natural transformation 
 \[(\Sigmas_\ctype(X+Y))^{X_\vtype}(\Gamma) = \mbox{$\Set$}^{\fset/{\Tyv}}\bigl({\fset/{\Tyv}}(\Gamma,
  \argument) \times X_\vtype, \Sigmas_\ctype(X+Y)\bigr)\]
 whose component at $\Delta\in \fset/\Tyv$ sends $(h,e)\in \fset/\Tyv(\Gamma,\Delta)\times X_\vtype(\Delta)$ to $f_\Delta(x_1,\ldots,x_{\dom(\Gamma)},e)$, where $x_i\in X_{\Gamma(i)}(\Delta)$ is the image of $i\in \dom(\Gamma)$ under the map
\[ V_{\Gamma(i)}(\Gamma)\xto{V_{\Gamma(i)}h} V_{\Gamma(i)}(\Delta) \xto{\var_{\Gamma(i),\Delta}} X_{\Gamma(i)}(\Delta).\]
The clause thus specifies the desired labeled transition 
$\mathsf{lam}_\kappa\, x\colon \varphi.t \xto{e} t[e/x]$ in the model \eqref{eq:op-model-cbpv}.
\end{remark}
}}{}

\end{document}